\documentclass[conference,letterpaper]{IEEEtran}
\IEEEoverridecommandlockouts
\pdfoutput=1
\usepackage{mathtools,amssymb,lipsum, nccmath}

\usepackage{cuted}
\usepackage{lipsum}
\usepackage{mathtools}
\usepackage{cuted}
\usepackage{verbatim}
\usepackage{cite}
\usepackage{amsmath,amssymb,amsfonts,amsthm,stmaryrd}
\usepackage{algorithmic}
\usepackage{graphicx}
\usepackage{tikz}
\usepackage{graphicx}
\usepackage{bbm}
\usepackage{bm}
\usepackage{textcomp}
\usepackage{xcolor}
\usepackage{graphicx}
\usepackage{pgfplots}
\usepackage{systeme}

\usetikzlibrary{patterns}
\usepackage{bigints}
\usepackage{color}
\usepgfplotslibrary{fillbetween}

\usepackage{soul}
\usepackage{textcomp}
\usepackage{xcolor}
\usepackage{subcaption}

\def\BibTeX{{\rm B\kern-.05em{\sc i\kern-.025em b}\kern-.08em
    T\kern-.1667em\lower.7ex\hbox{E}\kern-.125emX}}

\newtheorem{theorem}{Theorem}
\newtheorem{remark}{Remark}
\newtheorem{definition}{Definition}

\newtheorem{lemma}{Lemma}

\def\proof{\noindent{\it Proof}. \ignorespaces}
\def\endproof{\vbox{\hrule height0.6pt\hbox{\vrule height1.3ex%
width0.6pt\hskip0.8ex\vrule width0.6pt}\hrule height0.6pt}}

\newcommand{\N}{\mathbb{N}}

\newcommand{\R}{\mathbb{R}}
\newcommand{\E}{\mathbb{E}}

\usepackage[utf8]{inputenc} 
\usepackage[T1]{fontenc}
\usepackage{url}
\usepackage{ifthen}
\usepackage{cite}

\pgfplotsset{compat=1.17}

\begin{document}

\title{Strategic Communication with Cost-Dependent Decoders via the Gray-Wyner Network}

 
 \author{\IEEEauthorblockN{Rony Bou Rouphael and Ma\"{e}l Le Treust
\thanks{
Ma\"el Le Treust gratefully acknowledges financial support from INS2I CNRS, DIM-RFSI, SRV ENSEA, UFR-ST UCP, The Paris Seine Initiative and IEA Cergy-Pontoise. This research has been conducted as part of the project Labex MME-DII (ANR11-LBX-0023-01).} 
}\\
\IEEEauthorblockA{
ETIS UMR 8051, CY Cergy-Paris Université, ENSEA, CNRS,\\
6, avenue du Ponceau, 95014 Cergy-Pontoise CEDEX, FRANCE\\
Email: \{rony.bou-rouphael ; mael.le-treust\}@ensea.fr}\\
 }

\maketitle

\begin{abstract}
In decentralized decision-making problems, communicating agents choose their actions based on locally available information and knowledge about decision rules or strategies of other agents.  
In this work, we consider a strategic communication game between an informed and biased encoder and two decoders communicating via a Gray-Wyner network. All three agents are assumed to be rational and endowed with distinct objectives captured by non-aligned cost functions. The encoder selects and announces beforehand the compression scheme to be implemented. 
Then, it transmits three signals: a public signal, and a private signal to each decoder inducing a Bayesian game among the decoders. 
Our goal is to characterize the encoder's minimal long-run cost function subject to the optimal compression scheme that satisfies both decoders incentives constraints.



\end{abstract}



\section{Introduction} 

We study a decentralized decision-making problem with restricted communication between multiple agents with mismatched objectives. 
We are interested in designing an achievable coding scheme that minimizes the encoder's long run cost function subject to the challenges imposed by the Gray-Wyner network.  
This paper extends our work in \cite{rouphael2021strategic} to the case where decoders are cost-dependent, i.e. the cost function of each decoder depends on the action of the other decoder, and each decoder observes one private signal in addition to the public signal. Knowing the cost functions of all agents, the goal of each player is to minimize its respective cost.\\ Originally referred to as the sender-receiver game, the problem was formulated in the game theory literature with no restrictions on the amount of information transmitted. The Nash equilibrium solution of the cheap talk game was investigated by Crawford and Sobel in 
\cite{crawford1982}
. In \cite{KamenicaGentzkow11}, Kamenica and Gentzkow formulate the Stackelberg version of the strategic communication game
. This setting, referred to as the Bayesian persuasion game, is the one under study in this paper by considering a Gray-Wyner network with two decoders. The simple Gray-Wyner network was formulated in the seminal work \cite{graywynersimplenetwork} with a formulation of the region of attainable rates. The optimal region of second-order coding for the lossy Gray-Wyner network was derived in \cite{secondordergraywynerlossy}.

Information design with multiple designers interacting with a set of agents is studied in \cite{koesslerlaclau}.   
In \cite{sar1}, \cite{SaritasFurrerGeziciLinderYukselISIT2019}, the Nash equilibrium solution is investigated for multi-dimensional sources and quadratic cost functions, whereas the Stackelberg solution is studied in \cite{sar2}. The computational aspects of the persuasion game are considered in \cite{dughmi}. 
The strategic communication problem with a noisy channel is investigated in \cite{AkyolLangbortBasar15}, \cite{akyol2017information}, \cite{LeTreustTomala(Allerton)16}, \cite{jet}, and 
\cite{pointtopoint}. The case where the decoder privately observes a signal correlated to the state, also referred to as the Wyner-Ziv setting \cite{wyner-it-1976}, is studied in \cite{akyol2016role}, \cite{corsica2020} and \cite{LeTreustTomalaISIT21}. Vora and Kulkarni investigate the achievable rates for the strategic communication problem in \cite{voraachievable}, \cite{voraextraction} where the decoder is the Stackelberg leader.

In this paper, we investigate the two-receiver 
strategic communication via a Gray-Wyner network, as in Fig. \ref{fig:successive00}. We assume that the encoder $\mathcal{E}$ commits to and reveals an encoding strategy before observing the source. Each commitment of the encoder induces a Bayesian game among the decoders $\mathcal{D}_1$ and $\mathcal{D}_2$. Since the set of information policies is compact, the Bayesian game admits Bayes-Nash equilibria \cite{Yu1999EssentialEO}. We assume that decoders will select the pair of output sequences that minimizes their respective costs and maximizes the encoder's cost.
We characterize the encoder's optimal cost obtained with strategic Gray-Wyner lossy source coding 
that satisfies both decoders incentives constraints.

\begin{figure}
    \centering
       \begin{tikzpicture}[xscale=3,yscale=1.3]
\draw[thick,->](0.6,0.2)--(0.95,0.2);
\draw[thick,-](0.95,-0.58)--(1.25,-0.58)--(1.25,1)--(0.95,1)--(0.95,-0.58);
\node[below, black] at (1.1,1.425) {$c_e(U,V_1,V_2)$};
\node[below, black] at (1.1,0.4) {$\mathcal{E}$};
\draw[thick,-](1.25,0.2)--(2.4,0.2);
\draw[thick, ->](1.7,0.62)--(1.8,0.82);
\draw[thick, ->](1.7,-0.37)--(1.8,-0.18);
\draw[thick, ->](1.7,0.12)--(1.8,0.32);
\draw[thick, ->](1.25,0.7)--(2.25,0.7);
\draw[thick, <->](2.4,-0.01)--(2.4,0.43);
\draw[thick, ->](1.25,-0.3)--(2.25,-0.3);
\draw[thick,-](2.25,0.44)--(2.55,0.44)--(2.55,1)--(2.25,1)--(2.25,0.44);
\node[below, black] at (2.4,1.425) { $c_1(U,V_1,V_2)$};
\node[below, black] at (2.4,0.925) { $\mathcal{D}_{1}$};
\draw[thick,-](2.25,-0.58)--(2.55,-0.58)--(2.55,-0.02)--(2.25,-0.02)--(2.25,-0.58);
\node[below, black] at (2.4,-0.55) { $c_2(U,V_1,V_2)$};
\node[below, black] at (2.4,-0.15) { $\mathcal{D}_{2}$};
\draw[thick,->](2.55,0.7)--(2.85,0.7);
\draw[thick,->](2.55,-0.3)--(2.85,-0.3);
\node[above, black] at (2.72,-0.3) {$V^n_2$};
\node[above, black] at (2.72,0.68) {$V^n_1$};
\node[above, black] at (0.78,0.25) {$U^n$};
\node[above, black] at (1.5,0.2) {$M_0$};
\node[above, black] at (1.5,0.67){$M_1$};
\node[above, black] at (1.5,-0.3){$M_2$};
\node[above, black] at (1.85,0.3){$R_0$};
\node[above, black] at (1.85,-0.22){$R_2$};
\node[above, black] at (1.85,0.8){$R_1$};
    \end{tikzpicture}
    \caption{Gray-Wyner network with cost-dependent decoders.}
    \label{fig:successive00}
\end{figure}
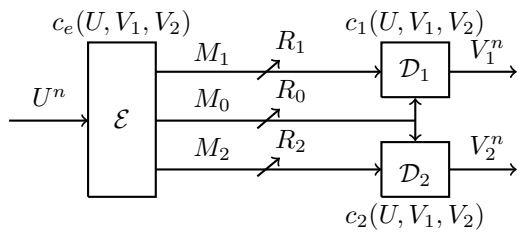


\subsection{Notations}

Let $\mathcal{E}$ denote the encoder and $\mathcal{D}_i$ denote decoder $i \in \{1,2\}$. Notations $U^n$ and $V^n_i$ denote the sequences of random variables of source information $u^n=(u_{1},...,u_{n}) \in \mathcal{U}^n $, and decoder $\mathcal{D}_i$'s actions $v^n_i \in \mathcal{V}^n_i$ respectively for $i \in \{1,2\}$. Calligraphic fonts $\mathcal{U}$ 
 and $\mathcal{V}_i$ denote the alphabets and lowercase letters $u$ 
 and $v_i$ denote the realizations. 
By a slight abuse of notation, indexing sets $\{1,...,2^{\lfloor nR \rfloor}\}, \ R \geq 0$ are referred to as message sets.  
For a discrete random variable $X,$ we denote by $\Delta(\mathcal{X})$ the probability simplex, i.e. the set of probability distributions over $\mathcal{X},$ and by $\mathcal{P}_{X}(x)$ the probability mass function $\mathbb{P}\{X=x\}$. Notation $X -\!\!\!\!\minuso\!\!\!\!- Y -\!\!\!\!\minuso\!\!\!\!- Z$ stands for the Markov chain $\mathcal{P}_{Z|XY}=\mathcal{P}_{Z|Y}$. The information source $U$ follows the independent and identically distributed (i.i.d) probability distribution $\mathcal{P}_{U}\in\Delta(\mathcal{U})$.

\section{System Model}
In this section, we introduce the coding problem 
and formulate the Bayesian game induced by each encoding function. 
\begin{definition}
Let $R_0,R_1,R_2 \in \mathbb{R}_{+}^3=[0,+\infty[^3$, and $n\in\mathbb{N}^{\star}=\mathbb{N}\backslash\{0\}$. The encoding function $\sigma$ and the decoding functions $\tau_i$ of the encoder $\mathcal{E}$ and decoders $\mathcal{D}_i$ for $i \in \{1,2\}$ respectively, are given by
\begin{align}
    \sigma:& \  \mathcal{U}^n \mapsto \Delta\big(\{1,..2^{\lfloor nR_0 \rfloor}\}\times\{1,..2^{\lfloor nR_1\rfloor}\}\times\{1,..2^{\lfloor nR_2\rfloor}\}\big), \nonumber\\
\tau_i:& \ (\{1,2,..2^{\lfloor nR_0 \rfloor}\}\times\{1,2,..2^{\lfloor nR_i\rfloor}\}) 
\mapsto \Delta(\mathcal{V}_i^n).\nonumber
 \end{align} 

\end{definition}
The coding triplets ($\sigma,\tau_1,\tau_2) 
$ are stochastic and induce a joint probability distribution 
\begin{align}
    &\mathcal{P}_{U^nM_0M_1M_2V_1^nV_2^n}^{\sigma,\tau_1,\tau_2} \nonumber = \\ & \bigg(\prod_{t=1}^n\mathcal{P}_{U_t}\bigg)\mathcal{P}^{\sigma}_{M_0M_1M_2|U^n}\mathcal{P}_{V_1^n|M_0M_1
    }^{\tau_1}\mathcal{P}_{V_2^n|M_0M_2
    }^{\tau_2}. 
\end{align}

\begin{definition} Single-letter cost functions $c_e:\mathcal{U}\times\mathcal{V}_1\times\mathcal{V}_2 \mapsto \mathbb{R}$ of the encoder and $c_i:\mathcal{U}\times\mathcal{V}_1\times\mathcal{V}_2  \mapsto \mathbb{R}$ of decoder $\mathcal{D}_i$ for $i \in \{1,2\}$ induce 
\label{def:longrun}
Long-run cost functions $c_e^n(\sigma,\tau_1,\tau_2)$ and $c_i^n(\sigma,\tau_1,\tau_2)$ 
\begin{align*}
&c_e^n(\sigma,\tau_1,\tau_2)= \mathbb{E}_{\sigma,\tau_1,\tau_2}\Bigg[\frac{1}{n}\sum_{t=1}^n c_e(U_t,V_{1,t},V_{2,t})\Bigg] \\
&=\sum_{u^n, v_1^n, v_2^n} \mathcal{P}_{U^nV_1^nV_2^n}^{\sigma,\tau_1,\tau_2}(u^n, v_1^n, v_2^n)\cdot\Bigg[\frac{1}{n}\sum_{t=1}^n c_e(u_t,v_{1,t},v_{2,t})\Bigg], \\
&c^n_{i}(\sigma,\tau_1,\tau_2)= \mathbb{E}_{\sigma,\tau_1,\tau_2}\Bigg[\frac{1}{n}\sum_{t=1}^n c_i(U_t,V_{1,t},V_{2,t})\Bigg]. 
\label{dn2} 
\end{align*}
where $\mathcal{P}^{\sigma,\tau_1,\tau_2}_{U^nV_1^nV_2^n}$ 
 denotes the marginal distributions of $\mathcal{P}^{\sigma,\tau_1,\tau_2}$ over the n-sequences $(U^n,V_1^n,V_2^n)$
 .
\end{definition}

 We consider the strategic communication game in which each player aims to minimize its long-run cost function. 
Each encoding function $\sigma$ induces a Bayesian game $G^{\sigma}(M_0,M_1,M_2,V_1^n,V_2^n)$ 
among the decoders which is defined below. 
 \begin{definition}
 For each encoding $\sigma$, the finite Bayesian game $G^{\sigma}(M_0,M_1,M_2,V_1^n,V_2^n)$
  \ consists of:
 \begin{itemize}
     \item the decoders $\mathcal{D}_i, i \in \{1,2\}$ as the players of the game,
     \item $\mathcal{V}_i^n$ is the set of action sequences of $\mathcal{D}_i$,
          \item  $(M_0,M_i)$ 
          is the type of decoder $\mathcal{D}_i$,
      \item $\tau_i$ 
is a behavior strategy of decoder $\mathcal{D}_i$,
 \item the belief of decoder $\mathcal{D}_1$ (resp. $\mathcal{D}_2$) over the type of decoder $\mathcal{D}_{2}$ (resp. $\mathcal{D}_1$) is given by $\mathcal{P}^{\sigma}_{M_2|M_0M_1}$ (resp. $\mathcal{P}^{\sigma}_{M_{1}|M_0M_2}$). 
     \item $C_i^{\sigma}: \{1,2,..2^{\lfloor nR_0 \rfloor}\}\times\{1,2,..2^{\lfloor nR_1\rfloor}\} \times\{1,2,..2^{\lfloor nR_2\rfloor}\}\times \mathcal{V}^n_1 \times \mathcal{V}_2^n \mapsto \mathbb{R} $ is the $\sigma$-cost function of $\mathcal{D}_i$ such that \begin{align} 
     C_i^{\sigma}(m_0,m_1,m_2,v_1^n,v_2^n)=\sum_{u^n}\mathcal{P}^{\sigma}(u^n|m_0,m_1,m_2)\times \nonumber \\ \Bigg[\frac{1}{n}\sum_{t=1}^n c_i(u_t,v_{1,t},v_{2,t})\Bigg], \ \forall v_1^n,v_2^n,m_0,m_1,m_2. \nonumber
 \end{align}
  \item for a fixed strategy profile $(\tau_1,\tau_2)$, the expected $\sigma$-costs $\Psi^{\sigma}_1(\tau_1,\tau_2,m_0,m_1)$ of $\mathcal{D}_1$ with type $(m_0,m_1)$ is given by
\begin{align}
&\Psi^{\sigma}_1(\tau_1,\tau_{2},m_0,m_1)=\sum_{m_{2}}\mathcal{P}^{\sigma}(m_{2}|m_0,m_1)\times \nonumber \\ &\sum_{v_1^n,v_{2}^n}\mathcal{P}^{\tau_1}(v_1^n|m_0,m_1)\mathcal{P}^{\tau_{2}}(v_{2}^n|m_0,m_{2})\times \nonumber \\&C_1^{\sigma}(v_1^n,v_{2}^n,m_0,m_1,m_{2}). \nonumber
\end{align}
Similarly, $\Psi^{\sigma}_2(\tau_1,\tau_2,m_0,m_2)$ can be defined.
 \end{itemize}
 \end{definition}


  \begin{definition} Given $\sigma$, 
for each behavior strategy $\tau_{2}$, decoder $\mathcal{D}_1$, 
computes the sets 
$BR_1^{\sigma}(\tau_{2})$  
of best-response strategies 
\begin{align}
    BR_1^{\sigma}(\tau_{2})=\big\{\tau_1, \Psi^{\sigma}_1(\tau_1,\tau_{2},m_0,m_1) \leq \Psi^{\sigma}_1(\tilde{\tau_1},\tau_{2},m_0,m_1), \nonumber \\ \ \forall \ \Tilde{\tau_1}, m_0,m_1 \big\}. \nonumber
\end{align}
Similarly, $\mathcal{D}_2$ computes $BR_2^{\sigma}(\tau_{1})$. 
\end{definition}



This Bayesian game $G^{\sigma}(M_0,M_1,M_2,V_1^n,V_2^n)$ is finite, the players use behavioral strategies and Nash Theorem \cite{nash1951non} ensures the existence of at least one Bayes-Nash equilibrium. In the following, we define the set of such equilibria.
\begin{definition}
For each encoding strategy $\sigma$, we define the set $BNE(\sigma)$ of Bayes-Nash equilibria $(\tau_1,\tau_{2})$ of $G^{\sigma}(M_0,M_1,M_2,V_1^n,V_2^n)$ as follows
\begin{equation}
    BNE(\sigma)=\{(\tau_1,\tau_{2}), \ \tau_1 \in BR_1^{\sigma}(\tau_2) \textrm{ and } \tau_{2} \in BR_2^{\sigma}(\tau_1)  
    \}. \nonumber
\end{equation}
\end{definition}

The communication game goes in the following order: 
 \begin{itemize}
\item The encoder $\mathcal{E}$ chooses, announces the encoding $\sigma$.
\item The sequence $U^n$ is drawn i.i.d with distribution $\mathcal{P}_U$, and the game $G^{\sigma}(M_0,M_1,M_2,V_1^n,V_2^n)$ begins.
\item The messages $(M_0,M_1,M_2)$ are encoded according to $\mathcal{P}^{\sigma}_{M_0M_1M_2|U^n}$.
\item Knowing $\sigma$, the decoders select the worst $BNE(\sigma)$ for the encoder's cost. 
    \item The cost values are given by $c_e^n(\sigma,\tau_1,\tau_2)$, $\Psi^{\sigma}_1(\tau_1,\tau_2,m_0,m_1)$, $\Psi^{\sigma}_2(\tau_1,\tau_2,m_0,m_2) $.
\end{itemize}

For $(R_0,R_1,R_2)\in \mathbb{R}_{+}^{3}$ and $n \in \mathbb{N}^{\star}$, the encoder has to solve the following coding problem.
\begin{equation}
\Gamma_e^n(R_0,R_1,R_2)=\underset{\sigma}{\inf}\underset{(\tau_1,\tau_2) \in BNE(\sigma), 
}{\max} c_e^n(\sigma,\tau_1,\tau_2).
\end{equation} 
\section{Main Result}
We consider three auxiliary random variables $W_0\in\mathcal{W}_0$, $W_1\in\mathcal{W}_1$ and $W_2\in\mathcal{W}_2$ with  $|\mathcal{W}_0|=|\mathcal{V}_1|\times|\mathcal{V}_2|+1$, and $|\mathcal{W}_i|=|\mathcal{V}_i|$, for $i \in \{1,2\}$. 
\begin{definition}\label{def:characterization}
For $(R_0,R_1,R_2) \in \mathbb{R}^{3}_{+}$, we define 
\begin{align}
\mathbb{Q}_0(R_0,&R_1,R_2) = \bigg\{ \mathcal{Q}_{W_0|U} \mathcal{Q}_{W_1|UW_0}\mathcal{Q}_{W_2|UW_0}, \nonumber \\ &R_0 \geq I(U;W_0), 
\ R_0+R_1 \geq I(U;W_1,W_0), \nonumber \\ 
&R_0+R_2\geq I(U;W_2,W_0) 
\bigg\},\label{srqq0} \\
\Hat{\mathbb{Q}}_0(R_0,&R_1,R_2) = \bigg\{ \mathcal{Q}_{W_0W_1W_2|U}, \nonumber \\ &R_0 \geq I(U;W_0), 
\ R_0+ R_1 \geq I(U;W_1,W_0), \nonumber \\ 
&R_0+R_2\geq I(U;W_2,W_0) 
\bigg\},\label{srqq0002}
\end{align}
In the following, we define the Bayesian game played among the decoders in the one-shot case scenario. We introduce the concept of decoders types in order to avoid hierarchy of Bayesian beliefs. 
\begin{definition} \label{def:singleletterBayesianGame}
For each distribution $\mathcal{Q}_{W_0W_1W_2|U}\in \Delta(\mathcal{W}_0\times\mathcal{W}_1\times \mathcal{W}_2)^{|\mathcal{U}|} $, the auxiliary single-letter Bayesian game $G^{w}(W_0,W_1,W_2,V_1,V_2)
$ 
is given as follows:
 \begin{itemize}
     \item  
     $(w_0,w_i)$ is the type of decoder $\mathcal{D}_i$, $i\in\{1,2\}$,
      \item the belief of decoder $\mathcal{D}_1$ (resp. $\mathcal{D}_2$) over the type of decoder $\mathcal{D}_{2}$ (resp. $\mathcal{D}_1$) is given by $\mathcal{Q}_{W_2|W_0W_1}$ (resp. $\mathcal{Q}_{W_{1}|W_0W_2}$).
     \item $C^{\star}_i:  \mathcal{V}_1 \times \mathcal{V}_2\times\mathcal{W}_0\times\mathcal{W}_1 \times \mathcal{W}_2 \mapsto \mathbb{R}$ is the single-letter cost of $\mathcal{D}_i$ such that $\forall v_1,v_2,w_0,w_1,w_2$ \begin{align} 
    C^{\star}_i(v_1,v_2,w_0,w_1,w_2)=\sum_{u}\mathcal{Q}(u|w_0,w_1,w_2)c_i(u,v_{1},v_{2}), \nonumber 
 \end{align} \vspace{-0.35cm}
\item  For each pair $(\mathcal{Q}_{V_1|W_0W_1},\mathcal{Q}_{V_2|W_0W_2})$ and profile $(w_0,w_i)$, 
the single-letter expected costs $\Psi^{\star}_i(\mathcal{Q}_{V_1|W_0W_1},\mathcal{Q}_{V_2|W_0W_2},w_0,w_i)$ of $\mathcal{D}_i$ are given by
\begin{align}
&\Psi^{\star}_1(\mathcal{Q}_{V_1|W_0W_1},\mathcal{Q}_{V_2|W_0W_2},w_0,w_1)= \nonumber \\ &\sum_{w_{2}}\mathcal{Q}(w_{2}|w_0,w_!) \sum_{v_1,v_2}\mathcal{Q}(v_1|w_0,w_1)\times \nonumber \\ &\mathcal{Q}(v_2|w_0,w_2)C^{\star}_1(v_1,v_2,w_0,w_1,w_2).  \nonumber
\end{align}
 \end{itemize}
 \end{definition}

Similarly, we get $\Psi^{\star}_2(\mathcal{Q}_{V_1|W_0W_1},\mathcal{Q}_{V_2|W_0W_2},w_0,w_2)$.

For each distribution $\mathcal{Q}_{W_0W_1W_2|U}\in \Delta(\mathcal{W}_0\times\mathcal{W}_1\times \mathcal{W}_2)^{|\mathcal{U}|} $, the auxiliary set of Bayes-Nash equilibria is given by
\begin{align}
\mathbb{BNE}&(\mathcal{Q}_{W_0W_1W_2|U})=\Big\{(\mathcal{Q}_{V_1| W_0W_1}, \mathcal{Q}_{V_2|W_0W_2}),   \nonumber \\ \quad  &\Psi^{\star}_1(\mathcal{Q}_{V_1|W_0W_1},\mathcal{Q}_{V_2|W_0W_2},w_0,w_1)\leq  \nonumber \\  &\Psi^{\star}_1(\tilde{\mathcal{Q}}_{V_1|W_0W_1},\mathcal{Q}_{V_2|W_0W_2},w_0,w_1) \ \forall \  \tilde{\mathcal{Q}}_{V_1|W_0W_1}, w_0,w_1,  \nonumber \\  \nonumber \quad &\Psi^{\star}_2(\mathcal{Q}_{V_1|W_0W_1},\mathcal{Q}_{V_2|W_0W_2},w_0,w_2)
\leq \nonumber \\ &\Psi^{\star}_2(\mathcal{Q}_{V_1|W_0W_1},\tilde{\mathcal{Q}}_{V_2|W_0W_2},w_0,w_2)\ \forall \  \tilde{\mathcal{Q}}_{V_2|W_0W_2}, w_0, w_2\Big\} . \nonumber
\end{align} 

The encoder's optimal cost is defined w.r.t. $\mathbb{Q}_0(R_0,R_1,R_2)$ and $\Hat{\mathbb{Q}}_0(R_0,R_1,R_2)$ respectively as follows
 \begin{align}
&\Gamma^{\star}_e(R_0,R_1,R_2)
= \nonumber \\ &\underset{\mathcal{Q}_{W_0|U}\mathcal{Q}_{W_1|W_0U}\atop\mathcal{Q}_{W_2|W_0U} 
\in\mathbb{Q}_0(R_0,R_1,R_2)}{\inf}\underset{(\mathcal{Q}_{V_1|W_0W_1},\mathcal{Q}_{V_2|W_0W_2}) \in \atop \mathbb{BNE}(\mathcal{Q}_{W_0|U}\mathcal{Q}_{W_1|W_0U}\mathcal{Q}_{W_2|W_0U} 
)}{\max}\hspace{-0.5cm}\mathbb{E} \Big[c_e(U,V_1,V_2)\Big], \label{optdistoooo1} \\
&\Hat{\Gamma}_e(R_0,R_1,R_2)
= \nonumber \\ &\underset{\mathcal{Q}_{W_0W_1W_2|U}\atop\in\Hat{\mathbb{Q}}_0(R_0,R_1,R_2)}{\inf}\underset{(\mathcal{Q}_{V_1|W_0W_1},\mathcal{Q}_{V_2|W_0W_2}) \in \atop \mathbb{BNE}(\mathcal{Q}_{W_0W_1W_2|U})}{\max}\mathbb{E} \Big[c_e(U,V_1,V_2)\Big], \label{optdistoooooo2}
\end{align} 
where the expectation in \eqref{optdistoooo1} is evaluated with respect to $\mathcal{P}_{U} \mathcal{Q}_{W_0|U}\mathcal{Q}_{W_1|W_0U}\mathcal{Q}_{W_2|UW_0}\mathcal{Q}_{V_1|W_0W_1}\mathcal{Q}_{V_2|W_0W_2}$, and satisfies the following Markov chain  \begin{align*}
     W_1 -\!\!\!\!\minuso\!\!\!\!- U,W_0 -\!\!\!\!\minuso\!\!\!\!- W_2.
 \end{align*} 
the expectation in \eqref{optdistoooooo2} is evaluated with respect to $\mathcal{P}_{U} \mathcal{Q}_{W_0W_1W_2|U}\mathcal{Q}_{V_1|W_0W_1}\mathcal{Q}_{V_2|W_0W_2}$. 
\end{definition}

\begin{remark}
The random variables $U, W_0,W_1,W_2,V_1$ and $V_2$ satisfy the following Markov chains
\begin{align*}
     &(U,W_2,V_2) -\!\!\!\!\minuso\!\!\!\!- (W_0,W_1) -\!\!\!\!\minuso\!\!\! \!- V_1, \\
     &(U,W_1,V_1) -\!\!\!\!\minuso\!\!\!\!- (W_0,W_2) -\!\!\!\!\minuso\!\!\! \!-V_2.
 \end{align*}
\end{remark}
\begin{theorem}\label{main result}
Let $(R_0,R_1,R_2) \in \mathbb{R}_{+}^3$, we have
\begin{align*} 
\forall \varepsilon>0,  \exists \hat{n} \in \mathbb{N},   \forall n \geq \hat{n}, \nonumber \\  \Gamma^n_e(R_0,R_1,R_2) &\leq \Gamma_e^{\star}(R_0,R_1,R_2) + \varepsilon,\\
\forall n \in \mathbb{N},  \Gamma_e^n(R_0,R_1,R_2) &\geq \Hat{\Gamma}_e(R_0,R_1,R_2).  
\end{align*}
\end{theorem}

\begin{lemma}\label{lemn} Let $(R_0,R_1,R_2) \in \mathbb{R}_{+}^3$, and consider $c_{e1}:\mathcal{U}\times\mathcal{V}_1\mapsto\mathbb{R}$, and $c_{e2}:\mathcal{U}\times\mathcal{V}_2\mapsto\mathbb{R}$. If for all $(u,v_1,v_2)$, $c_e(u,v_1,v_2)=c_{e1}(u,v_1)+c_{e2}(u,v_2)$, then
\begin{align}
     \Gamma_e^{\star}(R_0,R_1,R_2) = \Hat{\Gamma}_e(R_0,R_1,R_2)
\end{align}\end{lemma}
Using Fekete's Lemma for the sub-additive sequence $\big(n \Gamma_e^n(R_0,R_1,R_2)\big)_{n\in \N^{\star}}$ \cite[Lemma 1]{rouphael2021strategic} we get
$$\lim_{n\rightarrow \infty} \Gamma^n_e(R_0,R_1,R_2)=\inf \Gamma^n_e(R_0,R_1,R_2) = \Gamma^{\star}_e(R_0,R_1,R_2).$$

\section{Achievability Proof of Theorem \ref{main result}}
Our proof consists of three main parts: Firstly, we restrict the optimization to a dense subset of distributions inducing essential equilibria in order to ensure convergence. Secondly, we generate the codebook and show that the probability of error over the codebook is small. Thirdly, we outline the passage from the game by block to the single-letter game and we analyze Bayes-Nash equilibria of all intermediary Bayesian games. This is possible since the beliefs induced by the coding functions are close under the KL-divergence to the single-letter beliefs described using auxiliary random variables.

\subsection{Essential Equilibria}
The random coding scheme may induce some perturbations in the probability distribution $\mathcal{Q}_{W_0|U}\mathcal{Q}_{W_1|W_0U}\mathcal{Q}_{W_2|W_0U}$ of the Bayesian game of Definition \ref{def:singleletterBayesianGame}. A Bayesian game is essential \cite[Definition 4.1]{Yu1999EssentialEO} if small perturbations of the probability distributions may induce small perturbations of the set of Bayes-Nash equilibria. According to \cite[Theorem 4.2]{Yu1999EssentialEO}, the set of essential Bayesian games is a dense subset of the set of Bayesian games. 

\begin{definition}\label{defessential} Given $\mathcal{Q}_{W_0W_1W_2|U} \in \Delta(\mathcal{W}_0\times\mathcal{W}_1\times \mathcal{W}_2)^{|\mathcal{U}|}$, 
an equilibrium $(\mathcal{Q}_{V_1|W_1,W_0},\mathcal{Q}_{V_2|W_2,W_0}) \in \mathbb{BNE}(\mathcal{Q}_{W_0W_1W_2|U})$ is said to be \emph{essential} if for all $\varepsilon>0$, there exists an open neighborhood $\Omega$ of $\mathcal{Q}_{W_0W_1W_2|U}$ such that for all  $\Hat{\mathcal{Q}}_{W_0W_1W_2|U} 
\in \Omega$,
\begin{align} 
&(\Hat{\mathcal{Q}}_{V_1|W_1,W_0},\Hat{\mathcal{Q}}_{V_2|W_2,W_0}) \in \mathbb{BNE}(\Hat{\mathcal{Q}}_{W_0W_1W_2|U}
)
\Longrightarrow \nonumber \\ &||\mathcal{Q}_{V_1|W_0,W_1} - \Hat{\mathcal{Q}}_{V_1|W_0,W_1} ||+ ||\mathcal{Q}_{V_2|W_0,W_2} - \Hat{\mathcal{Q}}_{V_2|W_0,W_2}|| \leq \varepsilon. \nonumber
\end{align}
We denote by $\mathbb{EBNE}
(\mathcal{Q}_{W_0W_1W_2|U})$ the set of essential Bayes-Nash equilibria.
\end{definition}
\begin{definition}
For $(R_0,R_1,R_2) \in \mathbb{R}^{3}_{+}$, we define the set 
\begin{align}
&\tilde{\mathbb{Q}}_0(R_0,R_1,R_2) = \Big\{
\mathcal{Q}_{W_0|U}\mathcal{Q}_{W_1|W_0U}\mathcal{Q}_{W_2|W_0U},  \nonumber \\  &\underset{u,w_0,w_1,\atop w_2}{\min}\mathcal{Q}(w_0|u)\mathcal{Q}(w_1|w_0,u)\mathcal{Q}(w_2|w_0,u) >0, \nonumber \\ &R_0 > I(U;W_0), R_1 > I(U;W_1|W_0),R_2 > I(U;W_2|W_0),  \nonumber \\  
 &\mathbb{BNE}(\mathcal{Q}_{W_0|U}\mathcal{Q}_{W_1|W_0U}\mathcal{Q}_{W_2|W_0U})= \nonumber \\ &\hspace{2cm}\mathbb{EBNE}
(\mathcal{Q}_{W_0|U}\mathcal{Q}_{W_1|W_0U}\mathcal{Q}_{W_2|W_0U}) \Big\}.
\nonumber
    \end{align}

\end{definition}
In the following, we show that optimizing over the full set of target distributions results in the same cost as when the optimization is taken over the set of target distributions restricted to unique Nash Equilibrium.

\begin{definition}
We replace the set $\mathbb{Q}_0(R_0,R_1,R_2)$ by the set $\Tilde{\mathbb{Q}}_0(R_0,R_1,R_2)$ and we define the following program:
 \begin{align}
 &\tilde{\Gamma}_e(R_0,R_1,R_2)= \nonumber \\ &\underset{\mathcal{Q}_{W_0|U}\mathcal{Q}_{W_1|W_0U}\mathcal{Q}_{W_2|W_0U}
   \atop\in\Tilde{\mathbb{Q}}_0(R_0,R_1,R_2)}{\inf}\hspace{-0.2cm}\underset{(\mathcal{Q}_{V_1|W_0W_1},\mathcal{Q}_{V_2|W_0W_2})\in  \atop \mathbb{EBNE}
   (\mathcal{Q}_{W_0|U}\mathcal{Q}_{W_1|W_0U}\mathcal{Q}_{W_2|W_0U}
   )}{\max}\hspace{-0.8cm}\mathbb{E}\Big[c_e(U,V_1,V_2)\Big]. \nonumber
\end{align} 
\end{definition}

\begin{lemma}
\label{kklemmmaa5} For $(R_0,R_1,R_2) \in \mathbb{R}_{+}^3$, we have
\begin{align}
   \Gamma_e^{\star}(R_0,R_1,R_2) =  \tilde{\Gamma}_e(R_0,R_1,R_2). 
\end{align}
\end{lemma} 
\begin{proof}
The proof of Lemma \ref{kklemmmaa5} follows from  [Theorem 4.2, \cite{Yu1999EssentialEO}] and Lemmas   \ref{kklem5}, and \ref{lemzxc}.
\end{proof}

\begin{lemma} \label{kklem5} Given a distribution $\mathcal{Q}_{W_0W_1W_2|U} \in \Delta(\mathcal{W}_0\times\mathcal{W}_1\times\mathcal{W}_2)^{|\mathcal{U}|}$, we have
\begin{align}
    \underset{(\mathcal{Q}_{V_1|W_0W_1},\mathcal{Q}_{V_2|W_0W_2})\atop \in\mathbb{EBNE}(\mathcal{Q}_{W_0W_1W_2|U})}{\max}&\mathbb{E}\Big[c_e(U,V_1,V_2)\Big] = \nonumber \\ &\underset{(\mathcal{Q}_{V_1|W_0W_1},\mathcal{Q}_{V_2|W_0W_2})\atop\in\mathbb{BNE}(\mathcal{Q}_{W_0W_1W_2|U})}{\max}\mathbb{E}\Big[c_e(U,V_1,V_2)\Big].
\end{align} 
\end{lemma}
\begin{proof}{[Lemma \ref{kklem5}]}
Consider the upper semi-continuous \cite[Theorem 3.3]{Yu1999EssentialEO}] correspondence $$\mathcal{Q}_{W_0W_1W_2|U} \rightrightarrows \mathbb{BNE}(\mathcal{Q}_{W_0W_1W_2|U}).$$ Using  \cite[Theorem 2]{MKfort51}, the correspondence is lower semi-continuous, therefore continuous. Using Def. \ref{defessential} and Berge's Maximum Theorem \cite{Berge}, the result follows.  
\end{proof}

\begin{lemma} \label{lemzxc}
Let $(R_0,R_1,R_2) \in \R^3_{+}$, the set $\tilde{\mathbb{Q}}_0(R_0,R_1,R_2)$ is a dense subset of  $\mathbb{Q}_0(R_0,R_1,R_2)$. 
\end{lemma}
\begin{proof}{[Lemma \ref{lemzxc}]} Fix a distribution $\mathcal{Q}_{W_0W_1W_2|U} \in \Delta(\mathcal{W}_0\times\mathcal{W}_1\times\mathcal{W}_2)^{|\mathcal{U}|}$. 
The Bayes-Nash equilibrium correspondence $\mathcal{Q}_{W_0W_1W_2|U} \rightrightarrows \mathbb{BNE}(\mathcal{Q}_{W_0W_1W_2|U})$ is non-empty for finite games, compact valued and upper semi-continuous \cite[Theorems 3.2, 3.3]{Yu1999EssentialEO}. Therefore, every probability distribution in the set $\tilde{\mathbb{Q}}_0(R_0,R_1,R_2)$ induces an essential equilibria. Using \cite[Theorem 4.2]{Yu1999EssentialEO}, the result follows.
\end{proof}

\subsection{Codebook Generation}
Fix a conditional probability distribution $\mathcal{Q}_{W_0|U}\mathcal{Q}_{W_1|UW_0}\mathcal{Q}_{W_2|UW_0} \in\mathbb{Q}_0(R_0,R_1,R_2)$. There exists $\eta>0$ such that \begin{align}
    R_0 =& I(U;W_0) + \eta, \label{eqwwe}\\
    R_1 =& I(U;W_1|W_0) + \eta, \label{eqwwee}\\
    R_2 =& I(U;W_2|W_0) + \eta. \label{eqwweee}
\end{align} 

Randomly and independently generate $2^{\lfloor nR_{0}\rfloor}$ sequences $w_0^n(m_{0})$ for $m_{0} \in [1:2^{\lfloor{nR_{0}}\rfloor}]$, according to the i.i.d distribution $\Pi_{t=1}^n\mathcal{Q}_{W_0}(w_{0t})$. For each $(m_{1},m_{0}) \in [1:2^{\lfloor nR_{1} \rfloor}] \times [1:2^{\lfloor nR_{0} \rfloor}]$ generate a sequence 
$w_1^n(m_{1},m_{0})$ randomly and conditionally independently according to the i.i.d conditional distribution  $\Pi_{t=1}^n\mathcal{Q}_{W_1|W_0}(w_{1t}|w_{0t}(m_{0}))$. For each $(m_{2},m_{0}) \in [1:2^{\lfloor nR_{2} \rfloor}] \times [1:2^{\lfloor nR_{0} \rfloor}]$ generate a sequence 
$w_2^n(m_{2},m_{0})$ randomly and conditionally independently according to the i.i.d conditional distribution  $\Pi_{t=1}^n\mathcal{Q}_{W_2|W_0}(w_{2t}|w_{0t}(m_{0}))$. 
\\  
Coding algorithm: Encoder $\mathcal{E}$ observes $u^n$ and determines $m_0$ such that $(U^n, W^n_0(m_{0})) \in \mathcal{T}_{\delta}^n(\mathcal{P}_U\mathcal{Q}_{W_0|U})$,  $m_1$ such that $(U^n, W^n_0(m_{0}), W_1^n(m_{1},m_{0})) \in \mathcal{T}_{\delta}^n(\mathcal{P}_U\mathcal{Q}_{W_0W_1|U})$, and $m_2$ such that $(U^n, W^n_0(m_{0}), W_2^n(m_{2},m_{0})) \in \mathcal{T}_{\delta}^n(\mathcal{P}_U\mathcal{Q}_{W_0W_2|U})$. 
If such a jointly typical tuple doesn't exist, the source encoder sets $(m_0,m_{1},m_{2})$ to $(1,1,1)$.
Then, it sends  $(m_{0},m_1)$ to decoder $\mathcal{D}_{1}$, and $(m_{0},m_{2})$ to decoder $\mathcal{D}_{2}$ 
Decoder $\mathcal{D}_{1}$ declares $v_1^n$ and 
decoder $\mathcal{D}_{2}$ declares $v_2^n$ according to $\tau_1$ and $\tau_2$.
\subsection{Analysis of error probability}

We define the following error events
\begin{align}
    \mathcal{F}_{0} =& \{(U^n,W_0^n(m_{0})) \notin \mathcal{T}_{\delta}^n, \;  \forall m_{0}\}, \nonumber \\
    \forall m_0,\quad  \mathcal{F}_{1}(m_0) =& \{(U^n,W^n_0(m_{0}),W_1^n(m_{1},m_0)) \notin \mathcal{T}_{\delta}^n \  \forall m_{1}\},\nonumber \\
    \forall m_0,\quad  \mathcal{F}_{2}(m_0) =& \{(U^n,W^n_0(m_{0}),W_2^n(m_{2},m_0)) \notin \mathcal{T}_{\delta}^n \  \forall m_{2}\}.\nonumber
\end{align}

By the covering lemma \cite[Lemma 3.3]{elgamal},  $\mathcal{P}(\mathcal{F}_{0})$ tends to zero as $n \longrightarrow \infty$ if 
        $R_{0} \geq \  I(U;W_0) + \eta, \label{kkee1}$
    , $\mathcal{P}(\mathcal{F}_{1}(M_0)|\mathcal{F}_{0}^c)$ goes to zero by the covering lemma if 
        $   R_1  \geq \ I(U;W_1|W_0)  +\eta$, and $\mathcal{P}(\mathcal{F}_{2}(M_0)|\mathcal{F}_{0}^c)$ goes to zero by the covering lemma if 
        $   R_2  \geq \ I(U;W_2|W_0)  +\eta.$ 
 The expected probability of error over the codebook being small means that for all $\varepsilon_2 >0$, for all $\eta>0$, there exists $\Bar{\delta}>0$, for all $\delta\leq\Bar{\delta}$, there exists $\Bar{n} \in \mathbb{N}$ such that for all $n \geq \Bar{n}$ we have: \begin{align}
           \mathbb{E}\big[\mathcal{P}(\mathcal{F}_{0})\big] \leq& \varepsilon_2, \label{eqarr} \\ 
           \mathbb{E}\big[\mathcal{P}(\mathcal{F}_{1}(m_0)|\mathcal{F}_{0}^c)\big] \leq& \varepsilon_2,
           \label{eqarrr} \\
           \mathbb{E}\big[\mathcal{P}(\mathcal{F}_{2}(m_0)|\mathcal{F}_{0}^c)\big] \leq& \varepsilon_2,
           \label{eqarrrr} 
       \end{align}

\subsection{Analysis of Bayes-Nash Equilibria} 

In this section, we break down the Bayesian game $G^{\sigma}(M_0,M_1,M_2,V^n_{1},V^n_{2})$ played by blocks of $n$-sequences, into several games of stage $t$, for $t\in\{1,2,...,n\}$ that differ in the decoders' types and actions until we reach the single-letter game characterized with auxiliary random variables as illustrated in Fig.\ref{fig:chaingames}. We analyze the different Bayes-Nash equilibria of these Bayesian games. In order to do so, we need to control the beliefs of each decoder about the state and about the type of the other decoder.

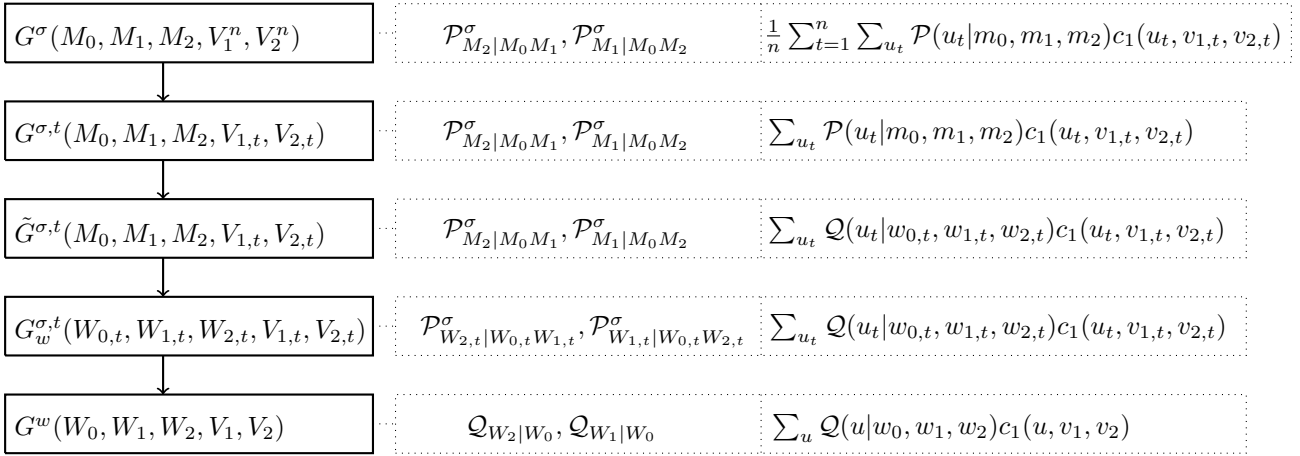
\begin{figure*}[ht]
    \centering
       \begin{tikzpicture}[xscale=3,yscale=1.3]
\draw[thick,-](0,0)--(1.63,0)--(1.63,0.6)--(0,0.6)--(0,0);
\node[right, black] at (0,0.25) {$G^{w}(W_{0},W_{1},W_{2},V_{1},V_{2})$};
\draw[thick,->](0.7,1)--(0.7,0.6);
\draw[dotted,-](1.63,0.3)--(1.73,0.3);
\draw[dotted,-](1.73,0)--(5.5,0)--(5.5,0.6)--(1.73,0.6)--(1.73,0);
\draw[dotted,-](3.35,0.6)--(3.35,0);
\node[right, black] at (2,0.25) {$\mathcal{Q}_{W_{2}|W_{0}},\mathcal{Q}_{W_{1}|W_{0}} \ \ \ \ \ \ \ \ \ \ \ \  \sum_{u}\mathcal{Q}(u|w_{0},w_{1},w_{2})c_1(u,v_{1},v_{2})$};
\draw[thick,-](0,1)--(1.63,1)--(1.63,1.6)--(0,1.6)--(0,1);
\node[right, black] at (0,1.25) {$G^{\sigma,t}_w(W_{0,t},W_{1,t},W_{2,t},V_{1,t},V_{2,t})$};
\draw[thick,->](0.7,2)--(0.7,1.6);
\draw[dotted,-](1.63,1.3)--(1.73,1.3);
\draw[dotted,-](1.73,1)--(5.5,1)--(5.5,1.6)--(1.73,1.6)--(1.73,1);
\draw[dotted,-](3.35,1.6)--(3.35,1);
\node[right, black] at (1.8,1.25) {$\mathcal{P}^{\sigma}_{W_{2,t}|W_{0,t}W_{1,t}},\mathcal{P}^{\sigma}_{W_{1,t}|W_{0,t}W_{2,t}} \ \ \sum_{u_t}\mathcal{Q}(u_t|w_{0,t},w_{1,t},w_{2,t})c_1(u_t,v_{1,t},v_{2,t})$};
\draw[thick,-](0,2)--(1.63,2)--(1.63,2.6)--(0,2.6)--(0,2);
\node[right, black] at (0,2.25) {$\Tilde{G}^{\sigma,t}(M_{0},M_{1},M_{2},V_{1,t},V_{2,t})$};
\draw[thick,->](0.7,3)--(0.7,2.6);
\draw[dotted,-](1.63,2.3)--(1.73,2.3);
\draw[dotted,-](1.73,2)--(5.5,2)--(5.5,2.6)--(1.73,2.6)--(1.73,2);
\draw[dotted,-](3.35,2.6)--(3.35,2);
\node[right, black] at (1.9,2.25) {$\mathcal{P}^{\sigma}_{M_{2}|M_{0}M_1},\mathcal{P}^{\sigma}_{M_{1}|M_{0}M_2} \ \ \ \ \ \ \ \ \sum_{u_t}\mathcal{Q}(u_t|w_{0,t},w_{1,t},w_{2,t})c_1(u_t,v_{1,t},v_{2,t})$};
\draw[thick,-](0,3)--(1.63,3)--(1.63,3.6)--(0,3.6)--(0,3);
\node[right, black] at (0,3.25) {$G^{\sigma,t}(M_{0},M_{1},M_{2},V_{1,t},V_{2,t})$};
\draw[thick,->](0.7,4)--(0.7,3.6);
\draw[dotted,-](1.63,3.3)--(1.73,3.3);
\draw[dotted,-](1.73,3)--(5.5,3)--(5.5,3.6)--(1.73,3.6)--(1.73,3);
\draw[dotted,-](3.35,3.6)--(3.35,3);
\node[right, black] at (1.9,3.25) {$\mathcal{P}^{\sigma}_{M_{2}|M_{0}M_1},\mathcal{P}^{\sigma}_{M_{1}|M_{0}M_2} \ \ \ \ \ \ \ \ \sum_{u_t}\mathcal{P}(u_t|m_0,m_1,m_2)c_1(u_t,v_{1,t},v_{2,t})$};
\draw[thick,-](0,4)--(1.63,4)--(1.63,4.6)--(0,4.6)--(0,4);
\node[right, black] at (0,4.25) {$G^{\sigma}(M_{0},M_{1},M_{2},V^n_{1},V^n_{2})$};
\draw[dotted,-](1.63,4.3)--(1.73,4.3);
\draw[dotted,-](1.73,4)--(5.7,4)--(5.7,4.6)--(1.73,4.6)--(1.73,4);
\draw[dotted,-](3.35,4.6)--(3.35,4);
\node[right, black] at (1.9,4.25) {$\mathcal{P}^{\sigma}_{M_{2}|M_{0}M_1},\mathcal{P}^{\sigma}_{M_{1}|M_{0}M_2} \ \ \ \ \ \ \ \ \frac{1}{n}\sum_{t=1}^n\sum_{u_t}\mathcal{P}(u_t|m_0,m_1,m_2)c_1(u_t,v_{1,t},v_{2,t})$};
    \end{tikzpicture}
    \caption{Chain of the Bayesian Games Played among the Decoders for Achievability.}
    \label{fig:chaingames}
\end{figure*}

\begin{definition}
 For each encoding $\sigma$, the finite Bayesian game  $G^{\sigma,t}(M_0,M_1,M_2,V_{1,t},V_{2,t})$ at stage $t$, for $t\in\{1,2,...n\}$
  \ consists of:
 \begin{itemize}
     \item the decoders $\mathcal{D}_i, i \in \{1,2\}$ as the players of the game,
     \item $\mathcal{V}_{i,t}$ is the set of action sequences of $\mathcal{D}_i$,
          \item  $(M_0,M_i)$ 
          is the type of decoder $\mathcal{D}_i$,
      \item $\tau_{i,t}: \ (\{1,2,..2^{\lfloor nR_0 \rfloor}\}\times\{1,2,..2^{\lfloor nR_i\rfloor}\}) 
\mapsto \Delta(\mathcal{V}_{i,t})$ 
is a behavior strategy of decoder $\mathcal{D}_i$,
 \item the belief of decoder $\mathcal{D}_1$ (resp. $\mathcal{D}_2$) over the type of decoder $\mathcal{D}_{2}$ (resp. $\mathcal{D}_1$) is given by $\mathcal{P}^{\sigma}_{M_2|M_0M_1}$ (resp. $\mathcal{P}^{\sigma}_{M_{1}|M_0M_2}$). 
     \item $C_i^{\sigma}: \{1,2,..2^{\lfloor nR_0 \rfloor}\}\times\{1,2,..2^{\lfloor nR_1\rfloor}\} \times\{1,2,..2^{\lfloor nR_2\rfloor}\}\times \mathcal{V}_{1,t} \times \mathcal{V}_{2,t} \mapsto \mathbb{R} $ is the $\sigma$-cost function of $\mathcal{D}_i$ at stage $t$ such that \begin{align} 
     &C_i^{\sigma}(m_0,m_1,m_2,v_{1,t},v_{2,t})=\sum_{u_t}\mathcal{P}^{\sigma}(u_t|m_0,m_1,m_2)\times \nonumber\\ &c_i(u_t,v_{1,t},v_{2,t}), \ \forall v_{1,t},v_{2,t},m_0,m_1,m_2. \nonumber
 \end{align}
  \item for a fixed strategy profile $(\tau_{1,t},\tau_{2,t})$, the expected $\sigma$-costs $\Psi^{\sigma}_1(\tau_{1,t},\tau_{2,t},m_0,m_1)$ of $\mathcal{D}_1$ at stage $t$ with type $(m_0,m_1)$ is given by
\begin{align}
&\Psi^{\sigma}_1(\tau_{1,t},\tau_{2,t},m_0,m_1)=\sum_{m_{2}}\mathcal{P}^{\sigma}(m_{2}|m_0,m_1)\times \nonumber \\ &\sum_{v_{1,t},v_{2,t}}\mathcal{P}^{\tau_{1,t}}(v_{1,t}|m_0,m_1)\mathcal{P}^{\tau_{2,t}}(v_{2,t}|m_0,m_{2})\times \nonumber \\&C_1^{\sigma}(v_{1,t},v_{2,t},m_0,m_1,m_{2}). \nonumber
\end{align}
Similarly, $\Psi^{\sigma}_2(\tau_{1,t},\tau_{2,t},m_0,m_2)$ can be defined.
 \end{itemize}
 \end{definition}

For each encoding strategy $\sigma$ and stage $t$, we define the set $BNE(\sigma,t)$ of Bayes-Nash equilibria $(\tau_{1,t},\tau_{2,t})$ of $G^{\sigma,t}(M_0,M_1,M_2,V_{1,t},V_{2,t})$ as follows
\begin{align}
    &BNE(\sigma,t)=\{(\tau_{1,t},\tau_{2,t}), \nonumber \\ &\Psi^{\sigma}_1(\tau_{1,t},\tau_{2,t},w_0,w_1)\leq  \Psi^{\sigma}_1(\tilde{\tau}_{1,t},\tau_{2,t},w_0,w_1) \forall \  \tilde{\tau}_{1,t}, w_0, w_1  \nonumber \\  \nonumber \quad &\Psi^{\sigma}_2(\tau_{1,t},\tau_{2,t},w_0,w_2)
\leq \Psi^{\sigma}_2(\tau_{1,t},\tilde{\tau}_{2,t},w_0,w_2) \forall \  \tilde{\tau}_{2,t}, w_0, w_2 
    \}. \nonumber
\end{align}
The following lemma shows that every Bayes-Nash equilibrium of the game $G^{\sigma}(M_0,M_1,M_2,V^n_{1},V^n_{2})$ played by blocks of $n$-sequences induces an equilibrium of the game  $G^{\sigma,t}(M_0,M_1,M_2,V_{1,t},V_{2,t})$ at stage $t$, for $t\in\{1,2,...n\}$. The converse is also true. 
\begin{lemma} \label{lemkjh}
1. If $(\tau_1,\tau_2) \in BNE(\sigma)$, then $(\tau_{1,t},\tau_{2,t}) \in BNE(\sigma,t)$ for all $t \in \{1,2,...,n\}$.\\
2. If $(\tau_{1,t},\tau_{2,t}) \in BNE(\sigma,t)$ for all $t \in \{1,2,...,n\}$, then $(\prod_{t=1}^n\tau_{1,t},\prod_{t=1}^n\tau_{2,t}) \in BNE(\sigma)$.
\end{lemma}
\begin{proof}
Given $\sigma$, let $(\tau_1,\tau_2) \in BNE(\sigma)$. Assume that there exists $t$ such that $(\tau_{1,t},\tau_{2,t}) \notin BNE(\sigma,t)$. This means that at $t$, at least one of the decoders is better off if it deviates from its strategy. Without loss of generality, assume $\mathcal{D}_1$ deviates to $\Tilde{\tau}_{1,t}$ and selects $\Tilde{v}_{1,t}$ accordingly. This shifts the action sequence $v_1^n$ that corresponds to $\tau_1$, to $\Tilde{v}_1^n=(v_{1,1},v_{1,2},...,\Tilde{v}_{1,t},...,v_{1,n})$. Thus $\tau_1 \notin BR^{\sigma}_1(\tau_2)$, and $(\tau_1,\tau_2) \notin BNE(\sigma)$.\\ 
    Conversely, if  $(\tau_{1,t},\tau_{2,t}) \in BNE(\sigma,t)$ for all $t \in \{1,2,...,n\}$, we define $(\tau_1,\tau_2)$ such that 
\begin{align}
    \mathcal{P}^{\tau_1}(v_1^n|m_0,m_1)=&\prod_{t=1}^n\mathcal{P}^{\tau_{1,t}}(v_{1,t}|m_{0},m_{1}), \quad \forall v_1^n,m_0,m_1 \\
     \mathcal{P}^{\tau_2}(v_2^n|m_0,m_2)=&\prod_{t=1}^n\mathcal{P}^{\tau_{2,t}}(v_{2,t}|m_{0},m_2), \quad \forall v_2^n,m_0,m_2.
\end{align}
Suppose that $(\tau_1,\tau_2) \notin BNE(\sigma)$. 
Without loss of generality, assume  $\tau_1 \notin BR_1^{\sigma}(\tau_2)$, i.e there exists $\Tilde{\tau}_1 \in BR_1(\sigma)$ such that $\Psi^{\sigma}_1(\tau_1,\tau_2,m_0,m_1) \geq \Psi^{\sigma}_1(\tilde{\tau_1},\tau_2,m_0,m_1)\ \forall m_0,m_1$. Therefore, there exists $t \in \{1,...,n\}$ such that $\Psi^{\sigma}_1(\tau_{1,t},\tau_{2,t},m_0,m_1) \geq \Psi^{\sigma}_1(\Tilde{\tau}_{1,t},\tau_{2,t},m_0,m_1) \ \forall m_0,m_1$. Thus, $(\tau_{1,t},\tau_{2,t}) \notin BNE(\sigma,t)$ which leads to the desired contradiction.  
\end{proof}

\begin{lemma} \label{corollary}
For all $\sigma$, we have
\begin{align}
  &\underset{(\tau_{1},\tau_{2}) \atop \in BNE(\sigma)}{\max}\mathbb{E}[\sum_{t=1}^n\frac{1}{n}c_e(U_t,V_{1,t},V_{2,t})] = \nonumber \\  &\frac{1}{n}\sum_{t=1}^n\underset{(\tau_{1,t},\tau_{2,t}) \atop \in BNE(\sigma,t)}{\max}\mathbb{E}[c_e(U_t,V_{1,t},V_{2,t})].
\end{align}
\end{lemma}
\begin{proof}
Let $\sigma$ be given. We will show the equality by showing double inequalities. Let $(\tau_1,\tau_2) \in BNE(\sigma)$ be arbitrarily chosen, and for all $t \in \{1,2,...,n\}$, let $(\tau_{1,t},\tau_{2,t}) \in BNE(\sigma,t)$ the corresponding equilibrium strategy pair for the game of stage $t$. Therefore, 
\begin{align}
  &\mathbb{E}_{\sigma,\tau_1,\tau_2}[\sum_{t=1}^n\frac{1}{n}c_e(U_t,V_{1,t},V_{2,t})] = \nonumber \\ 
   &\frac{1}{n}\sum_{t=1}^n\mathbb{E}_{\sigma,\tau_1,\tau_2}[c_e(U_t,V_{1,t},V_{2,t})] \leq \nonumber \\
   &\frac{1}{n}\sum_{t=1}^n\mathbb{E}_{\sigma,\tau_{1,t},\tau_{2,t}}[c_e(U_t,V_{1,t},V_{2,t})] \leq \label{www1} \\&\frac{1}{n}\sum_{t=1}^n\underset{(\tau_{1,t},\tau_{2,t}) \atop \in BNE(\sigma,t)}{\max}\mathbb{E}[c_e(U_t,V_{1,t},V_{2,t})].\nonumber
\end{align}
where \eqref{www1} follows since $(\tau_1,\tau_2) \in BNE(\sigma)$. Since this is true for all $(\tau_1,\tau_2) \in BNE(\sigma)$, then \begin{align}&\underset{(\tau_{1},\tau_{2}) \atop \in BNE(\sigma)}{\max}\mathbb{E}[\sum_{t=1}^n\frac{1}{n}c_e(U_t,V_{1,t},V_{2,t})]\nonumber \\ &\leq \frac{1}{n}\sum_{t=1}^n\underset{(\tau_{1,t},\tau_{2,t}) \atop \in BNE(\sigma,t)}{\max}\mathbb{E}[c_e(U_t,V_{1,t},V_{2,t})]. 
\end{align} Similarly, let $(\tau_{1,t},\tau_{2,t}) \in BNE(\sigma,t)$ be arbitrarily chosen for each $t \in \{1,2,...,n\}$. We have,
\begin{align}
    &\frac{1}{n}\cdot\sum_{t=1}^n\mathbb{E}_{\sigma,\tau_{1,t},\tau_{2,t}}[c_e(U_t,V_{1,t},V_{2,t})] \nonumber =\\&\frac{1}{n} \cdot\mathbb{E}_{\sigma,\tau_{1,t},\tau_{2,t}}[\sum_{t=1}^n c_e(U_t,V_{1,t},V_{2,t})] \nonumber \leq \\ 
    &\frac{1}{n}\mathbb{E}_{\sigma,\tau_1,\tau_2}[\sum_{t=1}^n c_e(U_t,V_{1,t},V_{2,t})] \nonumber \leq \label{www2}\\
    &\underset{(\tau_{1},\tau_{2}) \atop \in BNE(\sigma)}{\max}\mathbb{E}[\sum_{t=1}^n\frac{1}{n}c_e(U_t,V_{1,t},V_{2,t})].
\end{align}
where \eqref{www2} follows since $(\tau_{1,t},\tau_{2,t}) \in BNE(\sigma,t)$ for all $t \in \{1,2,...,n\}$. Since this is true for all  $(\tau_{1,t},\tau_{2,t}) \in BNE(\sigma,t)$, then $\underset{(\tau_{1},\tau_{2}) \atop \in BNE(\sigma)}{\max}\mathbb{E}[\sum_{t=1}^n\frac{1}{n}c_e(U_t,V_{1,t},V_{2,t})] \geq \frac{1}{n}\sum_{t=1}^n\underset{(\tau_{1,t},\tau_{2,t}) \atop \in BNE(\sigma,t)}{\max}\mathbb{E}[c_e(U_t,V_{1,t},V_{2,t})].$ This concludes the proof.
\end{proof}

We introduce the indicator of error events $E_{\delta} \in \{0,1\}$ defined as follows
\begin{align}\label{yyttyytt}
  E_{\delta}=&\begin{cases}
    1, & \text{if $(u^n,w_1^n,w_2^n,w_0^n) \notin \mathcal{T}_{\delta}^n
    $}.\\
    0, & \text{otherwise}. 
  \end{cases} 
\end{align}
We control the Bayesian belief of decoder $\mathcal{D}_1$ (resp. $\mathcal{D}_2$) about the type of $\mathcal{D}_2$ (resp. $\mathcal{D}_1$). Let $\mathcal{P}_{W_{2,t}}^{w_0^n,w_1^n} \in \Delta(\mathcal{W}_2)$ denote $\mathcal{P}_{W_{2,t}|W_0^n,W_1^n}(.|w_0^n,w_1^n)$ and $\mathcal{P}_{W_{1,t}}^{w_0^n,w_2^n} \in \Delta(\mathcal{W}_1)$ denote $\mathcal{P}_{W_{1,t}|W_0^n,W_2^n}(.|w_0^n,w_2^n)$. In a similar fashion, we denote by $\mathcal{Q}_{W_{2}}^{w_{0},w_{1}}$ and $\mathcal{Q}_{W_{1}}^{w_{0},w_{2}}$ the distributions $\mathcal{Q}_{W_{2,t}|W_{0,t},W_{1,t}}(.|w_{0},w_{1})$ and $\mathcal{Q}_{W_{1,t}|W_{0,t},W_{2,t}}(.|w_{0},w_{2})$ respectively.
\begin{lemma}\label{lemm134} For all $w^n_0,w^n_1,w_2^n,w_0,w_1,w_2$, we have
\begin{align}
&\lim_{n\mapsto \infty}\mathbb{E} \Big[\frac{1}{n}\sum_{t=1}^n D(\mathcal{P}_{W_{2,t}}^{w_0^n,w_1^n} || \mathcal{Q}_{W_{2}}^{w_{0},w_{1}}) \Big|E_{\delta}=0\Big] =0,\label{lkjlkj}\\
&\lim_{n\mapsto \infty}\mathbb{E} \Big[\frac{1}{n}\sum_{t=1}^n D(\mathcal{P}_{W_{1,t}}^{w^n_0,w_2^n}|| \mathcal{Q}_{W_{1}}^{w_{0},w_{2}}) \Big|E_{\delta}=0\Big] =0.
\end{align}
\end{lemma}
\begin{proof}
The proof of Lemma \ref{lemm134} is stated in App. \ref{battod1}. 
\end{proof}

We denote the Bayesian posterior beliefs $\mathcal{P}^{\sigma}_{U_t|M_1M_2M_0}(\cdot|m_1,m_2,m_0)\in\Delta(\mathcal{U})$ 
by $\mathcal{P}^{m_1,m_2,m_0}_{U_t}$ 
, and by $\mathcal{Q}_{U}^{w_{1}w_2w_{0}}$ 
the single-letter belief $\mathcal{Q}_{U|W_{1}W_2W_{0}}(\cdot|w_1,w_2,w_0)$. 
\begin{lemma}\label{lemm133}
For all $m_0,m_1,m_2,w_0,w_1,w_2$ 
, we have
\begin{align}
   &\lim_{n\mapsto \infty}\mathbb{E} \Big[\frac{1}{n}\sum_{t=1}^n D(\mathcal{P}_{U_t}^{m_0m_1m_2} ||\mathcal{Q}_{U}^{w_{1}w_2w_{0}}) \Big|E_{\delta}=0\Big] 
=0.
\end{align}
\end{lemma}
\begin{proof}
The proof of Lemma \ref{lemm133} is stated in App. \ref{cobatsu}. 
\end{proof}
\begin{lemma}\label{lemmtre}
For all $w^n_{0},w^n_{1},w^n_2$
, we have
\begin{align}
   &\lim_{n\mapsto \infty}\mathbb{E} \Big[\frac{1}{n}\sum_{t=1}^n D(\mathcal{P}^{\sigma}_{W_{1,t}|W_{0,t}W_{2,t}} ||\mathcal{Q}_{W_{1,t}|W_{2,t}W_{0,t}}) \Big|E_{\delta}=0\Big] \nonumber \\
&=0, \nonumber \\
   &\lim_{n\mapsto \infty}\mathbb{E} \Big[\frac{1}{n}\sum_{t=1}^n D(\mathcal{P}^{\sigma}_{W_{2,t}|W_{0,t}W_{1,t}} ||\mathcal{Q}_{W_{2,t}|W_{1,t}W_{0,t}}) \Big|E_{\delta}=0\Big] \nonumber \\
&=0. \nonumber
\end{align}
\end{lemma}
\begin{proof}
The proof of lemma \ref{lemmtre} is stated in App. \ref{cobatru}.
\end{proof}

In the following we define the game $\Tilde{G}^{\sigma,t}(M_0,M_1,M_2,V_{1,t},V_{2,t})$ of stage $t$, for $t\in\{1,2,...n\}$, in which the costs are computed using the single-letter beliefs that correspond to essential equilibria.
\begin{definition}
 For each encoding $\sigma$, the Bayesian game  $\Tilde{G}^{\sigma,t}(M_0,M_1,M_2,V_{1,t},V_{2,t})$ of stage $t$, for $t\in\{1,2,...n\}$
  \ consists of:
 \begin{itemize}
     \item the decoders $\mathcal{D}_i, i \in \{1,2\}$ as the players of the game,
     \item $\mathcal{V}_{i,t}$ is the set of action sequences of $\mathcal{D}_i$,
          \item  $(M_0,M_i)$ 
          is the type of decoder $\mathcal{D}_i$,
      \item $\tau_{i,t}: \ (\{1,2,..2^{\lfloor nR_0 \rfloor}\}\times\{1,2,..2^{\lfloor nR_i\rfloor}\}) 
\mapsto \Delta(\mathcal{V}_{i,t})$ 
is a behavior strategy of decoder $\mathcal{D}_i$,
 \item the belief of decoder $\mathcal{D}_1$ (resp. $\mathcal{D}_2$) over the type of decoder $\mathcal{D}_{2}$ (resp. $\mathcal{D}_1$) is given by $\mathcal{P}^{\sigma}_{M_2|M_0M_1}$ (resp. $\mathcal{P}^{\sigma}_{M_{1}|M_0M_2}$). 
     \item $\Tilde{C}_i^{\sigma}: \{1,2,..2^{\lfloor nR_0 \rfloor}\}\times\{1,2,..2^{\lfloor nR_1\rfloor}\} \times\{1,2,..2^{\lfloor nR_2\rfloor}\}\times \mathcal{V}_{1,t} \times \mathcal{V}_{2,t} \mapsto \mathbb{R} $ is the $\sigma$-cost function of $\mathcal{D}_i$ at stage $t$ such that \begin{align} 
     &\Tilde{C}_i^{\sigma}(m_0,m_1,m_2,v_{1,t},v_{2,t}) \nonumber \\ &=\sum_{u_t}\mathcal{Q}(u_t|w_{0,t}(m_0),w_{1,t}(m_0,m_1),w_{2,t}(m_0,m_2))\times \nonumber\\ &c_i(u_t,v_{1,t},v_{2,t}), \ \forall v_{1,t},v_{2,t},m_0,m_1,m_2. \nonumber
 \end{align}
  \item for a fixed strategy profile $(\tau_{1,t},\tau_{2,t})$, the expected $\sigma$-costs $\Psi^{\sigma}_1(\tau_{1,t},\tau_{2,t},m_0,m_1)$ of $\mathcal{D}_1$ at stage $t$ with type $(m_0,m_1)$ is given by
\begin{align}
&\Tilde{\Psi}^{\sigma}_1(\tau_{1,t},\tau_{2,t},m_0,m_1)=\sum_{m_{2}}\mathcal{P}^{\sigma}(m_{2}|m_0,m_1)\times \nonumber \\ &\sum_{v_{1,t},v_{2,t}}\mathcal{P}^{\tau_{1,t}}(v_{1,t}|m_0,m_1)\mathcal{P}^{\tau_{2,t}}(v_{2,t}|m_0,m_{2})\times \nonumber \\&\Tilde{C}_1^{\sigma}(v_{1,t},v_{2,t},m_0,m_1,m_{2}). \nonumber
\end{align}
Similarly, $\Tilde{\Psi}^{\sigma}_2(\tau_{1,t},\tau_{2,t},m_0,m_2)$ can be defined.
 \end{itemize}
 \end{definition}
For each distribution $\mathcal{P}^{\sigma}_{W_{0,t}W_{1,t}W_{2,t}|U}\in \Delta(\mathcal{W}_{0,t}\times\mathcal{W}_{1,t}\times \mathcal{W}_{2,t})^{|\mathcal{U}_t|} $, the set of essential Bayes-Nash equilibria of stage $t$ is given by

\begin{align}
&\hspace{-0.5cm}\Tilde{BNE}(\sigma,t)=\Big\{(\tau_{1,t},\tau_{2,t}),    \nonumber \\ \quad  &\hspace{-0.5cm}\Tilde{\Psi}^{\sigma}_1(\tau_{1,t},\tau_{2,t},m_0,m_1)\leq  \Tilde{\Psi}^{\sigma}_1(\tilde{\tau}_{1,t},\tau_{2,t},m_0,m_1) \forall \tilde{\tau}_{1,t}, m_0, m_1  \nonumber \\  \nonumber \quad &\hspace{-0.5cm}\Tilde{\Psi}^{\sigma}_2(\tau_{1,t},\tau_{2,t},m_0,m_2)
\leq \Tilde{\Psi}^{\sigma}_2(\tau_{1,t},\tilde{\tau}_{2,t},m_0,m_2)\forall \tilde{\tau}_{2,t}, m_0, m_2 \Big\}. \nonumber
\end{align} 

\begin{lemma}\label{ineqttzz00}
For all $\sigma$, and $\varepsilon>0$, we have \begin{align}
   &\sum_{t=1}^n\frac{1}{n} \bigg|\underset{(\tau_{1,t},\tau_{2,t}) \atop \in BNE(\sigma,t)}{\max}\mathbb{E}[c_e(U_t,V_{1,t},V_{2,t})] - \nonumber \\ &\underset{(\tau_{1,t},\tau_{2,t}) \atop \in \Tilde{BNE}(\sigma,t)}{\max}\mathbb{E}[c_e(U_t,V_{1,t},V_{2,t})] \bigg| \leq \varepsilon. \label{ineqttzz}
\end{align} 
\end{lemma}
\begin{proof} Consider the following correspondence \begin{align}
    &\mathcal{P}^{\sigma}_{U_t|M_0M_1M_2} \rightrightarrows \bigg\{(\mathcal{P}^{\tau_{1,t}}_{V_{1,t}|M_0M_1},\mathcal{P}^{\tau_{2,t}}_{V_{2,t}|M_0M_2}),  \nonumber \\ \quad  &\Tilde{\Psi}^{\sigma}_1(\tau_{1,t},\tau_{2,t},m_0,m_1)\leq  \Tilde{\Psi}^{\sigma}_1(\tilde{\tau}_{1,t},\tau_{2,t},m_0,m_1) \forall \tilde{\tau}_{1,t}, m_0, m_1  \nonumber \\   \quad &\Tilde{\Psi}^{\sigma}_2(\tau_{1,t},\tau_{2,t},m_0,m_2)
\leq \Tilde{\Psi}^{\sigma}_2(\tau_{1,t},\tilde{\tau}_{2,t},m_0,m_2)\forall \tilde{\tau}_{2,t}, m_0, m_2  \bigg\}.\label{edfedf}\end{align} Denote by $N(\mathcal{P}^{\sigma}_{U_t|M_0M_1M_2})$ the RHS of equation \eqref{edfedf}.
It follows from lemma \ref{lemm133} that for a given $\varepsilon>0$, we have, \begin{align} 
     & \sum_{t=1}^n\frac{1}{n}\bigg|C_i^{\sigma}(m_0,m_1,m_2,v_{1,t},v_{2,t}) - \Tilde{C}_i^{\sigma}(m_0,m_1,m_2,v_{1,t},v_{2,t})\bigg|= \nonumber \\
     & \sum_{t=1}^n\frac{1}{n}\bigg|\sum_{u_t}\mathcal{P}^{\sigma}(u_t|m_0,m_1,m_2)c_i(u_t,v_{1,t},v_{2,t})- \nonumber \\ &\sum_{u_t}\mathcal{Q}(u_t|w_{0,t}(m_0),w_{1,t}(m_0,m_1),w_{2,t}(m_0,m_2))\times \nonumber\\ &c_i(u_t,v_{1,t},v_{2,t})\bigg| \leq \varepsilon, \quad
     \forall m_0,m_1,m_2,v_{1,t},v_{2,t} . 
 \end{align}
 Consequently, for all $m_0,m_1,m_2,v_{1,t},v_{2,t}$,\begin{align}
   & \sum_{t=1}^n\frac{1}{n}\bigg|\Psi^{\sigma}_1(\tau_{1,t},\tau_{2,t},m_0,m_1) - \Tilde{\Psi}^{\sigma}_1(\tau_{1,t},\tau_{2,t},m_0,m_1)\bigg| \nonumber \\  &= \sum_{t=1}^n\frac{1}{n}\bigg|\sum_{m_{2}}\mathcal{P}^{\sigma}(m_{2}|m_0,m_1) \sum_{v_{1,t},v_{2,t}}\mathcal{P}^{\tau_{1,t}}(v_{1,t}|m_0,m_1)\times \nonumber\\&\mathcal{P}^{\tau_{2,t}}(v_{2,t}|m_0,m_{2})C_1^{\sigma}(v_{1,t},v_{2,t},m_0,m_1,m_{2}) - \nonumber\\ &\sum_{m_{2}}\mathcal{P}^{\sigma}(m_{2}|m_0,m_1)\sum_{v_{1,t},v_{2,t}}\mathcal{P}^{\tau_{1,t}}(v_{1,t}|m_0,m_1)\times\nonumber \\ &\mathcal{P}^{\tau_{2,t}}(v_{2,t}|m_0,m_{2})\Tilde{C}_1^{\sigma}(v_{1,t},v_{2,t},m_0,m_1,m_{2})\bigg| \leq \varepsilon.
 \end{align} 
Therefore, using \cite[Theorem 2]{MKfort51},  the correspondence in \eqref{edfedf} is continuous. 
 Using Berge's Maximum Theorem \cite{Berge}, the function \begin{align}
     \mathcal{P}^{\sigma}_{U_t|M_0M_1M_2} \mapsto \underset{(\tau_{1,t},\tau_{2,t}) \atop \in N(\mathcal{P}^{\sigma}_{U_t|M_0M_1M_2})}{\max}\mathbb{E}[c_e(U_t,V_{1,t},V_{2,t})] \label{maxvaluefunction}
 \end{align}
 is well-defined and continuous. Hence, $(\tau_{1,t},\tau_{2,t}) \in \Tilde{BNE}(\sigma,t)$. Therefore, varying $\sigma$ in a small neighborhood, slightly perturbs the expected cost functions resulting in a slightly perturbed set of Bayes-Nash equilibria. By the continuity of the max-value function in \eqref{maxvaluefunction}, we get the desired inequality. 
\end{proof}

Similarly, denote by $G_w^{\sigma,t}(W_{0,t},W_{1,t},W_{2,t},V_{1,t},V_{2,t})$ the Bayesian game restricted to types $W_{0,t},W_{1,t},W_{2,t}$ at stage $t$, for $t\in\{1,2,...n\}$ defined as follows.

\begin{definition}
 For each encoding $\sigma$, the finite Bayesian game  $G_w^{\sigma,t}(W_{0,t},W_{1,t},W_{2,t},V_{1,t},V_{2,t})$ at stage $t$, for $t\in\{1,2,...n\}$
  \ consists of:
 \begin{itemize}
     \item the decoders $\mathcal{D}_i, i \in \{1,2\}$ as the players of the game,
     \item $\mathcal{V}_{i,t}$ is the set of action sequences of $\mathcal{D}_i$,
          \item  $(W_{0,t},W_{i,t})$ 
          is the type of decoder $\mathcal{D}_i$,
      \item $\tau_{i,t}: \ (\{1,2,..2^{\lfloor nR_0 \rfloor}\}\times\{1,2,..2^{\lfloor nR_i\rfloor}\}) 
\mapsto \Delta(\mathcal{V}_{i,t})$ 
is a behavior strategy of decoder $\mathcal{D}_i$,
 \item the belief of decoder $\mathcal{D}_1$ (resp. $\mathcal{D}_2$) over the type of decoder $\mathcal{D}_{2}$ (resp. $\mathcal{D}_1$) is given by $\mathcal{P}^{\sigma}_{W_{2,t}|W_{0,t}W_{1,t}}$ (resp. $\mathcal{P}^{\sigma}_{W_{1,t}|W_{0,t}W_{2,t}}$). 
     \item $C_i^{\sigma,w}: \mathcal{W}_0 \times\mathcal{W}_1\times\mathcal{W}_2\times \mathcal{V}_{1,t} \times \mathcal{V}_{2,t} \mapsto \mathbb{R} $ is the $\sigma$-cost function of $\mathcal{D}_i$ at stage $t$ such that \begin{align} 
     &C_i^{\sigma,w}(w_{0,t},w_{1,t},w_{2,t},v_{1,t},v_{2,t})=\nonumber\\ &\sum_{u_t}\mathcal{Q}(u_t|w_{0,t},w_{1,t},w_{2,t})c_i(u_t,v_{1,t},v_{2,t}), \nonumber\\ &\qquad \forall v_{1,t},v_{2,t},w_{0,t},w_{1,t},w_{2,t}. \nonumber
 \end{align}
  \item for a fixed strategy profile $(\tau_{1,t},\tau_{2,t})$, the expected $\sigma$-costs $\Psi^{\sigma,w}_1(\tau_{1,t},\tau_{2,t},w_{0,t},w_{1,t})$ of $\mathcal{D}_1$ at stage $t$ with type $(w_{0,t},w_{1,t})$ is given by
\begin{align}
&\Psi^{\sigma,w}_1(\tau_{1,t},\tau_{2,t},w_{0,t},w_{1,t})=\sum_{w_{2,t}}\mathcal{P}^{\sigma}(w_{2,t}|w_{0,t},w_{1,t})\times \nonumber \\ &\sum_{v_{1,t},v_{2,t}}\mathcal{P}^{\tau_{1,t}}(v_{1,t}|w_{0,t},w_{1,t})\mathcal{P}^{\tau_{2,t}}(v_{2,t}|w_{0,t},w_{2,t})\times \nonumber \\&C_1^{\sigma,w}(v_{1,t},v_{2,t},w_{0,t},w_{1,t},w_{2,t}). \nonumber
\end{align}
Similarly, $\Psi^{\sigma,w}_2(\tau_{1,t},\tau_{2,t},w_{0,t},w_{2,t})$ can be defined.
 \end{itemize}
 \end{definition}

For each encoding strategy $\sigma$ and stage $t$, we define the set $BNE(\sigma,t)$ of Bayes-Nash equilibria $(\tau_{1,t},\tau_{2,t})$ of $G_w^{\sigma,t}(W_{0,t},W_{1,t},W_{2,t},V_{1,t},V_{2,t})$ as follows
\begin{align}
    &BNE^w(\sigma,t)=\Big\{(\tau_{1,t},\tau_{2,t}), \nonumber \\ &\Psi^{\sigma,w}_1(\tau_{1,t},\tau_{2,t},w_0,w_1)\leq  \Psi^{\sigma,w}_1(\tilde{\tau}_{1,t},\tau_{2,t},w_0,w_1) \nonumber \\ &\qquad \forall  \tilde{\tau}_{1,t}, w_0, w_1  \nonumber \\  \nonumber \quad &\Psi^{\sigma,w}_2(\tau_{1,t},\tau_{2,t},w_0,w_2)
\leq \Psi^{\sigma,w}_2(\tau_{1,t},\tilde{\tau}_{2,t},w_0,w_2) \nonumber \\ &\qquad\forall \tilde{\tau}_{2,t}, w_0, w_2 \Big\}
    . \nonumber
\end{align}

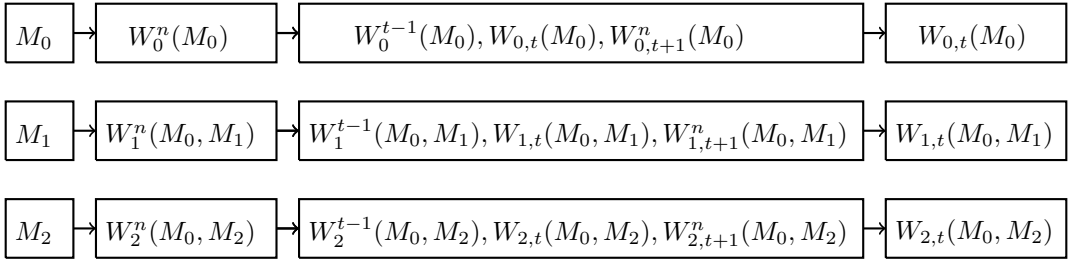
\begin{figure*}[ht]
    \centering
       \begin{tikzpicture}[xscale=3,yscale=1.3]
\draw[thick,-](0,0)--(0.3,0)--(0.3,0.6)--(0,0.6)--(0,0);
\node[right, black] at (0,0.25) {$M_2$};
\draw[thick,->](0.3,0.3)--(0.4,0.3);
\draw[thick,-](0.4,0)--(1.2,0)--(1.2,0.6)--(0.4,0.6)--(0.4,0);
\node[right, black] at (0.4,0.25) {$W_2^n(M_0,M_2)$};
\draw[thick,->](1.2,0.3)--(1.3,0.3);
\draw[thick,-](1.3,0)--(3.8,0)--(3.8,0.6)--(1.3,0.6)--(1.3,0);
\node[right, black] at (1.3,0.25) {$W_{2}^{t-1}(M_0,M_2),W_{2,t}(M_0,M_2),W_{2,t+1}^n(M_0,M_2)$};
\draw[thick,->](1.2,0.3)--(1.3,0.3);
\draw[thick,->](3.8,0.3)--(3.9,0.3);
\draw[thick,-](3.9,0)--(4.7,0)--(4.7,0.6)--(3.9,0.6)--(3.9,0);
\node[right, black] at (3.9,0.25) {$W_{2,t}(M_0,M_2)$};
\draw[thick,-](0,1)--(0.3,1)--(0.3,1.6)--(0,1.6)--(0,1);
\node[right, black] at (0,1.25) {$M_1$};
\draw[thick,->](0.3,1.3)--(0.4,1.3);
\draw[thick,-](0.4,1)--(1.2,1)--(1.2,1.6)--(0.4,1.6)--(0.4,1);
\node[right, black] at (0.4,1.25) {$W_1^n(M_0,M_1)$};
\draw[thick,->](1.2,1.3)--(1.3,1.3);
\draw[thick,-](1.3,1)--(3.8,1)--(3.8,1.6)--(1.3,1.6)--(1.3,1);
\node[right, black] at (1.3,1.25) {$W_{1}^{t-1}(M_0,M_1),W_{1,t}(M_0,M_1),W_{1,t+1}^n(M_0,M_1)$};
\draw[thick,->](1.2,1.3)--(1.3,1.3);
\draw[thick,->](3.8,1.3)--(3.9,1.3);
\draw[thick,-](3.9,1)--(4.7,1)--(4.7,1.6)--(3.9,1.6)--(3.9,1);
\node[right, black] at (3.9,1.25) {$W_{1,t}(M_0,M_1)$};
\draw[thick,-](0,2)--(0.3,2)--(0.3,2.6)--(0,2.6)--(0,2);
\node[right, black] at (0,2.25) {$M_0$};
\draw[thick,->](0.3,2.3)--(0.4,2.3);
\draw[thick,-](0.4,2)--(1.2,2)--(1.2,2.6)--(0.4,2.6)--(0.4,2);
\node[right, black] at (0.5,2.25) {$W_0^n(M_0)$};
\draw[thick,->](1.2,2.3)--(1.3,2.3);
\draw[thick,-](1.3,2)--(3.8,2)--(3.8,2.6)--(1.3,2.6)--(1.3,2);
\node[right, black] at (1.5,2.25) {$W_{0}^{t-1}(M_0),W_{0,t}(M_0),W_{0,t+1}^n(M_0)$};
\draw[thick,->](1.2,2.3)--(1.3,2.3);
\draw[thick,->](3.8,2.3)--(3.9,2.3);
\draw[thick,-](3.9,2)--(4.7,2)--(4.7,2.6)--(3.9,2.6)--(3.9,2);
\node[right, black] at (4,2.25) {$W_{0,t}(M_0)$};
    \end{tikzpicture}
    \caption{Marginalization Scheme of the Decoders' Types}
    \label{fig:marginals}
\end{figure*}

The game $G_w^{\sigma,t}(W_{0,t},W_{1,t},W_{2,t},V_{1,t},V_{2,t})$ 
of stage $t$ directly derives from $\Tilde{G}^{\sigma,t}(M_0,M_1,M_2,V_{1,t},V_{2,t})$ by marginalizing with respect to components $W_{0,t}$, $W_{1,t}$ and $W_{2,t}$ of $W_0^n(M_0), \ W_1^n(M_0,M_1)$ and $W_2^n(M_0,M_2)$ respectively according to the injections in Fig. \ref{fig:marginals}.
\begin{lemma}\label{eqttzz00}
For all $\sigma$, we have
\begin{align}
    \underset{(\tau_{1,t},\tau_{2,t}) \atop \in \Tilde{BNE}(\sigma,t)}{\max}\mathbb{E}[c_e(U_t,V_{1,t},V_{2,t})] = \underset{(\tau_{1,t},\tau_{2,t}) \atop \in  BNE^w(\sigma,t)}{\max}\mathbb{E}[c_e(U_t,V_{1,t},V_{2,t})] \label{eqttzz}
\end{align} 
\end{lemma}
\begin{proof}We will proceed by showing that every equilibrium $(\tau_{1,t},\tau_{2,t})\in \Tilde{BNE}(\sigma,t)$ induces an equilibrium in $BNE^w(\sigma,t)$. For all $\sigma,\tau_{1,t},\tau_{2,t},w_{0,t},w_{1,t}$, we have,
\begin{align}
&\Psi^{\sigma,w}_1(\tau_{1,t},\tau_{2,t},w_{0,t},w_{1,t})=\sum_{w_{2,t}}\mathcal{P}^{\sigma}(w_{2,t}|w_{0,t},w_{1,t})\times \nonumber \\ &\sum_{v_{1,t},v_{2,t}}\mathcal{P}^{\tau_{1,t}}(v_{1,t}|w_{0,t},w_{1,t})\mathcal{P}^{\tau_{2,t}}(v_{2,t}|w_{0,t},w_{2,t})\times \nonumber \\&\sum_{u_t}\mathcal{Q}(u_t|w_{0,t},w_{1,t},w_{2,t})c_i(u_t,v_{1,t},v_{2,t}) =  
\nonumber \\ &\sum_{w^n_{2},w_{0,t+1}^n,w^{t-1}_{0}\atop w^{t-1}_1,w_{1,t+1}^n }\mathcal{P}^{\sigma}(w^n_{2}|w^n_{0},w^n_{1})\times \nonumber \\&\mathcal{P}^{\sigma}(w_{0,t+1}^n,w^{t-1}_{0},w^{t-1}_1,w_{1,t+1}^n|w_{0,t},w_{1,t})\times \nonumber \\ &\sum_{v_{1,t},v_{2,t}}\sum_{w_{0,t+1}^n,w^{t-1}_{0}\atop w^{t-1}_1,w_{1,t+1}^n}\mathcal{P}^{\tau_{1,t}}(v_{1,t}|w_0^n,w_1^n)\times \nonumber \\ &\mathcal{P}(w_{0,t+1}^n,w^{t-1}_{0},w^{t-1}_1,w_{1,t+1}^n|w_{0,t},w_{1,t})\times \nonumber \\&\sum_{w_{0,t+1}^n,w^{t-1}_{0}\atop w^{t-1}_2,w_{2,t+1}^n}\mathcal{P}^{\tau_{2,t}}(v_{2,t}|w_0^n,w_2^n)\times\nonumber \\&\mathcal{P}(w_{0,t+1}^n,w^{t-1}_{0}, w^{t-1}_2,w_{2,t+1}^n|w_{0,t},w_{2,t})\times \nonumber \\&\sum_{u_t,w_{0,t+1}^n,w^{t-1}_{0}\atop w^{t-1}_1,w_{1,t+1}^n, w^{t-1}_2,w_{2,t+1}^n}\mathcal{Q}(u_t|w^n_{0},w^n_{1},w^n_{2})\times\nonumber\\&\mathcal{Q}(w_{0,t+1}^n,w^{t-1}_{0},w^{t-1}_1,w_{1,t+1}^n, w^{t-1}_2,w_{2,t+1}^n|w_{0,t},w_{1,t},w_{2,t})\times \nonumber \\ &c_1(u_t,v_{1,t},v_{2,t}).
\end{align} Thus, if $(\tau_{1,t},\tau_{2,t})\in BNE(\sigma,t)$, then $(\tau_{1,t},\tau_{2,t})\in BNE^w(\sigma,t)$. 
\end{proof}

Finally, we show that the encoder's cost in the essential game $G^{w}(W_{0},W_{1},W_{2},V_{1},V_{2})$ is close to its cost in the game $G_{w}^{\sigma,t}(W_{0,t},W_{1,t},W_{2,t},V_{1,t},V_{2,t})$ of stage $t$, for $t\in\{1,2,...n\}$.
\begin{lemma} \label{lazlez}
For all $\sigma$, $\mathcal{Q}_{W_{0,t}|U}\mathcal{Q}_{W_{1,t}|W_{0,t}U}\mathcal{Q}_{W_{2,t}|W_{0,t}U}\in\mathbb{Q}_0(R_0,R_1,R_2)$ and $\varepsilon>0$, we have 
\begin{align}
     \sum_{t=1}^n\frac{1}{n}\Bigg|&\underset{(\tau_{1,t},\tau_{2,t})\atop \in BNE^w(\sigma,t)}{\max}\mathbb{E}[c_e(U_t,V_{1,t},V_{2,t})] \ - \nonumber \\ &\underset{(\mathcal{Q}_{V_{1,t}|W_{0,t}W_{1,t}},\mathcal{Q}_{V_{2,t}|W_{0,t}W_{2,t}})\in \atop  \mathbb{BNE}(\mathcal{Q}_{W_{0,t}|U_t}\mathcal{Q}_{W_{1,t}|W_{0,t}U_t}\mathcal{Q}_{W_{2,t}|W_{0,t}U_t})}{\max}\mathbb{E}[c_e(U_t,V_{1,t},V_{2,t})] \Bigg| \nonumber \\ &\leq \varepsilon. \nonumber
\end{align}
\end{lemma}
\begin{proof}
Consider the following correspondence \begin{align}
    &\mathcal{P}^{\sigma}_{U_t|W_{0,t}W_{1,t}W_{2,t}} \rightrightarrows \bigg\{(\mathcal{P}^{\tau_{1,t}}_{V_{1,t}|W_{0,t}W_{1,t}},\mathcal{P}^{\tau_{2,t}}_{V_{2,t}|W_{0,t}W_{2,t}}),  \nonumber \\ \quad  &\Tilde{\Psi}^{w,t}_1(\tau_{1,t},\tau_{2,t},w_{0,t},w_{1,t})\leq  \Tilde{\Psi}^{w,t}_1(\tilde{\tau}_{1,t},\tau_{2,t},w_{0,t},w_{1,t}) \nonumber \\ &\forall \tilde{\tau}_{1,t}, w_{0,t},w_{1,t},  \nonumber \\   \quad &\Tilde{\Psi}^{w,t}_2(\tau_{1,t},\tau_{2,t},w_{0,t},w_{2,t})
\leq \Tilde{\Psi}^{w,t}_2(\tau_{1,t},\tilde{\tau}_{2,t},w_{0,t},w_{2,t}) \nonumber \\&\forall \tilde{\tau}_{2,t}, w_{0,t},w_{2,t} \bigg\}.\label{edfedf2}\end{align} Denote by $N^w(\mathcal{P}^{\sigma}_{U_t|W_{0,t}W_{1,t}W_{2,t}})$ the RHS of equation \eqref{edfedf}.
It follows from lemma \ref{lemmtre} that for a given $\varepsilon>0$, and for all $w_{0,t},w_{1,t},w_{2,t},v_{1,t},v_{2,t}$,\begin{align}
   & \sum_{t=1}^n\frac{1}{n}\bigg|\Psi^{w,t}_1(\tau_{1,t},\tau_{2,t},w_{0,t},w_{1,t}) - \Tilde{\Psi}^{w,t}_1(\tau_{1,t},\tau_{2,t},w_{0,t},w_{1,t})\bigg| \nonumber \\  &= \sum_{t=1}^n\frac{1}{n}\bigg|\sum_{w_{2,t}}\mathcal{P}^{\sigma}(w_{2,t}|w_{0,t},w_{1,t}) \sum_{v_{1,t},v_{2,t}}\mathcal{P}^{\tau_{1,t}}(v_{1,t}|w_{0,t},w_{1,t})\times \nonumber\\&\mathcal{P}^{\tau_{2,t}}(v_{2,t}|w_{0,t},w_{2,t})C_1^{w,t}(v_{1,t},v_{2,t},w_{0,t},w_{1,t},w_{2,t}) - \nonumber\\ &\sum_{w_{2,t}}\mathcal{Q}(w_{2,t}|w_{0,t})\sum_{v_{1,t},v_{2,t}}\mathcal{Q}^{\tau_{1,t}}(v_{1,t}|w_{0,t},w_{1,t})\times\nonumber \\ &\mathcal{Q}^{\tau_{2,t}}(v_{2,t}|w_{0,t},w_{2,t})\Tilde{C}_1^{w,t}(v_{1,t},v_{2,t},w_{0,t},w_{1,t},w_{2,t})\bigg| \leq \varepsilon.
 \end{align} 
Therefore, using \cite[Theorem 2]{MKfort51},  the correspondence in \eqref{edfedf2} is continuous. 
 Using Berge's Maximum Theorem \cite{Berge}, the function \begin{align}
     \mathcal{P}^{\sigma}_{U_t|W_{0,t}W_{1,t}W_{2,t}} \mapsto \underset{(\tau_{1,t},\tau_{2,t}) \atop \in N^w(\mathcal{P}^{\sigma}_{U_t|W_{0,t}W_{1,t}W_{2,t}})}{\max}\mathbb{E}[c_e(U_t,V_{1,t},V_{2,t})] \label{maxvaluefunction2}
 \end{align}
 is well-defined and continuous. Hence, $(\tau_{1,t},\tau_{2,t}) \in \Tilde{BNE}(\sigma,t)$. Therefore, varying $\sigma$ in a small neighborhood, slightly perturbs the expected cost functions resulting in a slightly perturbed set of Bayes-Nash equilibria. By the continuity of the max-value function in \eqref{maxvaluefunction2}, we get the desired inequality. 
\end{proof} \\
 It follows from lemmas \ref{corollary}, \ref{ineqttzz00}, \ref{eqttzz00}, and \ref{lazlez}
that for all $\varepsilon>0$, there exists $\Hat{n}$ such that for all $n\geq\Hat{n}$, 
 \begin{align}
       &\underset{(\tau_{1},\tau_{2}) \atop \in BNE(\sigma)}{\max}\mathbb{E}[\sum_{t=1}^n\frac{1}{n}c_e(U_t,V_{1,t},V_{2,t})] = \nonumber \\  &\frac{1}{n}\sum_{t=1}^n\underset{(\tau_{1,t},\tau_{2,t}) \atop \in BNE(\sigma,t)}{\max}\mathbb{E}[c_e(U_t,V_{1,t},V_{2,t})] \leq\nonumber \\ 
      &\frac{1}{n}\sum_{t=1}^n\underset{(\tau_{1,t},\tau_{2,t}) \atop \in \Tilde{BNE}(\sigma,t)}{\max}\mathbb{E}[c_e(U_t,V_{1,t},V_{2,t})] \nonumber \leq \\
        &\frac{1}{n}\sum_{t=1}^n\underset{(\tau_{1,t},\tau_{2,t}) \atop \in \Tilde{BNE}(\sigma,t)}{\max}\mathbb{E}[c_e(U_t,V_{1,t},V_{2,t})] \nonumber = \\ &\frac{1}{n}\sum_{t=1}^n\underset{(\tau_{1,t},\tau_{2,t}) \atop \in  BNE^w(\sigma,t)}{\max}\mathbb{E}[c_e(U_t,V_{1,t},V_{2,t})] \nonumber \leq \\
        &\underset{(\mathcal{Q}_{V_{1,t}|W_{0,t}W_{1,t}},\mathcal{Q}_{V_{2,t}|W_{0,t}W_{2,t}})\in \atop  \mathbb{BNE}(\mathcal{Q}_{W_{0,t}|U_t}\mathcal{Q}_{W_{1,t}|W_{0,t}U_t}\mathcal{Q}_{W_{2,t}|W_{0,t}U_t})}{\max}\mathbb{E}[c_e(U_t,V_{1,t},V_{2,t})] + \varepsilon.
 \end{align}

This concludes the achievability proof of Theorem \ref{main result}. 
\endproof{}


\section{Converse proof of Theorem \ref{main result}}
Let $(R_0,R_1,R_2) \in \mathbb{R}^3_{+}$ and $n\in\N^{\star}$. Fix $(\sigma,\tau_1,\tau_2)
$, and consider a random variable $T$ uniformly distributed over $\{1,2,...,n\}$ and independent of $(U^n,M_0,M_1,M_2,V_1^n,V_2^n)$. We introduce the auxiliary random variables $W_0 =(M_0,T)$, $W_1 =M_1$, $W_2 =M_2$, $(U,V_1,V_2)=(U_T,V_{1,T},V_{2,T})$, 
 distributed according to $\mathcal{P}_{UW_0W_1W_2V_1V_2}^{\sigma\tau_1\tau_2}$ defined for all $(u,w_0,w_1,w_2,v_1,v_2) = (u_t,m_0,m_1,m_2,t,v_{1,t},v_{2,t})$ by
\begin{align*}
  &\mathcal{P}_{UW_0W_1W_2V_1V_2}^{\sigma\tau_1\tau_2} (u,w_0,w_1,w_2,v_1,v_2)
  = \nonumber \\ &\mathcal{P}_{U_TM_0M_1M_2TV_{1T}V_{2T}}^{\sigma\tau_1\tau_2}  (u_t,m_0,m_1,m_2,t,v_{1,t},v_{2,t})\\
 =&\frac1n\sum_{u^{t-1}\atop u_{t+1}^n}\sum_{v_1^{t-1},v_{1,t+1}^n\atop v_2^{t-1},v_{2,t+1}^n}\!\!\!\!\!\!
\bigg(\prod_{t=1}^n\mathcal{P}_{U}(u_t)\bigg)\mathcal{P}^{\sigma}_{M_0M_1M_2|U^n}(m_0,m_1,m_2|u^n) \nonumber \\ \times &\mathcal{P}^{\tau_1}_{V_1^n|M_0M_1}(v_1^n|m_0,m_1)\mathcal{P}^{\tau_2}_{V_2^n|M_0M_2}(v_2^n|m_0,m_2). 
\end{align*}
\begin{lemma}\label{lemma:decomposition}
The distribution $\mathcal{P}_{UW_0W_1W_2V_1V_2}^{\sigma\tau_1\tau_2}$ has marginal on $\Delta(\mathcal{U})$ given by $\mathcal{P}_U$ and satisfies the Markov chain properties
\begin{align*}
     (U,V_2)   -\!\!\!\!\minuso\!\!\!\!-  (W_0,W_1)  -\!\!\!\!\minuso\!\!\!\!- V_1; \\
  (U,W_1,V_1)  -\!\!\!\!\minuso\!\!\!\!-  (W_0,W_2)  -\!\!\!\!\minuso\!\!\!\!- V_2.
 \end{align*}
\end{lemma}
\begin{proof}[Lemma \ref{lemma:decomposition}]
The i.i.d. property of the source ensures that the marginal distribution is $\mathcal{P}_U$. By the definition of the decoding functions $\tau_1$ and $\tau_2$ we have 
\begin{align*}
     &(U_T,V_{2,T})   -\!\!\!\!\minuso\!\!\!\!-  (M_1,M_0,T)  -\!\!\!\!\minuso\!\!\!\!- V_{1,T},\\
     &(U_T,M_1,V_{1,T})  -\!\!\!\!\minuso\!\!\!\!-  (M_2,M_0,T)  -\!\!\!\!\minuso\!\!\!\!- V_{2,T}.
 \end{align*}
\end{proof}
Therefore $\mathcal{P}_{UW_0W_1W_2V_1V_2}^{\sigma\tau_1\tau_2} = \mathcal{P}_U\mathcal{P}_{W_0|U}^{\sigma} \mathcal{P}_{W_1|W_0U}^{\sigma}\mathcal{P}_{W_2|W_0U}^{\sigma}\mathcal{P}_{V_1|W_0W_1}^{\tau_1}\mathcal{P}_{V_2|W_0W_2}^{\tau_2}$.

\begin{lemma}\label{lemma:belongtoQ0}
For all $\sigma$, the distribution $\mathcal{P}_{W_0W_1W_2|U}^{\sigma}\in \mathbb{Q}_0(R_0,R_1,R_2)$.
\end{lemma}
\begin{proof}[Lemma \ref{lemma:belongtoQ0}]
We consider an encoding strategy $\sigma$, then 
\begin{align}
n R_0 \geq& H(M_0) \geq I(M_0;U^n)  \label{e:Ide1} \\
=& \sum_{t=1}^n I(U_t;M_0|U^{t-1}) \label{memorylesssourceyyy} \\
=& n I(U_T;M_0|U^{T-1},T)\label{e:Ide2}\\
=& n I(U_T;M_0,U^{T-1},T)\label{e:Ide3}\\
\geq& n I(U_T;M_0,T)\label{e:Ide4}\\
=& n I(U;W_0).\label{e:Ide5}
\end{align}
In fact, \eqref{e:Ide2} follows from the introduction of the uniform random variable $T\in\{1,\ldots,n\}$, \eqref{e:Ide3} comes from the i.i.d. property of the source and \eqref{e:Ide5} follows from the identification of the auxiliary random variables $(U,W_2)$. Similarly,
\begin{align} 
nR_1 \geq&  H(M_1)  \geq I(U^n;M_1|M_0) \label{1lossy source} \\ 
\geq& nI(U_T;M_1|M_0,T) \\ 
=&nI(U;W_1|W_0).
\end{align}
Similarly, 
$nR_2 \geq  
nI(U;W_2|W_0).$
\end{proof}

\begin{lemma}\label{lemma:singlelettercost}
For all $(\sigma,\tau_1,\tau_2)$ and $i\in\{1,2\}$, 
we have
\begin{align}
   c_e^n(\sigma,\tau_1,\tau_2) =& \E   \big[c_e(U,V_1,V_2)\big],\\
   c_i^n(\sigma,\tau_1,\tau_{2}) =& \E
   \big[c_i(U,V_1,V_{2})\big].
\end{align}
evaluated with respect to $\mathcal{P}_U\mathcal{P}_{W_0W_1W_2|U}^{\sigma}\mathcal{P}_{V_1|W_0W_1}^{\tau_1}\mathcal{P}_{V_2|W_0W_2}^{\tau_2}$.

Moreover, for each $(m_0,m_1,m_2,v_1^n,v_2^n)$, we have 
\begin{align}
&C_i^{\sigma}(m_0,m_1,m_2,v_1^n,v_2^n)=C^{\star}_i(w_0,w_1,w_2,v_1,v_2), \\
  &\Psi^{\sigma}_i(\tau_1,\tau_{2},m_0,m_i) =  \E_{\mathcal{P}_{U}}
  \big[\Psi^{\star}_i(\mathcal{P}^{\tau_1}_{V_1|W_0W_1},\mathcal{P}^{\tau_2}_{V_2|W_0W_2},w_0,w_i)\big].\nonumber
\end{align}

\end{lemma}

\begin{proof}[Lemma \ref{lemma:singlelettercost}] By Definition \ref{def:longrun} we have
\begin{align}
&c_e^n(\sigma,\tau_1,\tau_2) = \sum_{u^n, m_0,m_1,m_2,\atop v_1^n,v_2^n}\bigg(\prod_{t=1}^n\mathcal{P}_{U}(u_t)\bigg) \nonumber \\
&\times\mathcal{P}^{\sigma}_{M_0M_1M_2|U^n}(m_0,m_1,m_2|u^n)\mathcal{P}^{\tau_1}_{V_1^n|M_0M_1}(v_1^n|m_0,m_1) \nonumber \\
&\times\mathcal{P}^{\tau_2}_{V_2^n|M_0M_2}(v_2^n|m_0,m_2)\cdot\Bigg[\frac{1}{n}\sum_{t=1}^n c_e(u_t,v_{1,t},v_{2,t})\Bigg] \nonumber\\
=& \sum_{t=1}^n \sum_{u_t,m_0,m_1,m_2,\atop v_{1,t}, v_{2,t}} \mathcal{P}^{\sigma,\tau_1,\tau_2}(u_t,m_0,m_1,m_2,t, v_{1,t}, v_{2,t}) \nonumber \\
&\times c_e(u_t,v_{1,t},v_{2,t}) = \E   \big[c_e(U,V_1,V_2)\big] .\nonumber
\end{align} For all $(m_0,m_1,m_2,v_1^n,v_2^n)$ and $i\in\{1,2\}$ we have
\begin{align}
&C_i^{\sigma}(m_0,m_1,m_2,v_1^n,v_2^n)= \nonumber \\ &\sum_{u^n}\mathcal{P}^{\sigma}_{U^n|M_0M_1M_2|}(u^n|m_0,m_1,m_2)\Bigg[\frac{1}{n}\sum_{t=1}^n c_i(u_t,v_{1,t},v_{2,t})\Bigg] \nonumber \\ 
&=\sum_{t=1}^n\sum_{u_t}\mathcal{P}^{\sigma}(u_t|m_0,m_1,m_2,t)c_i(u_t,v_{1,t},v_{2,t}) \nonumber \\
&=\sum_{u}\mathcal{P}^{\sigma}(u|w_0,w_1,w_2)c_i(u,v_{1},v_{2}) \nonumber \\
&=C^{\star}_i(v_1,v_2,w_0,w_1,w_2).\nonumber
\end{align} Moreover,
\begin{align}
&\Psi^{\sigma}_1(\tau_1,\tau_{2},m_0,m_1)=\sum_{u^n}\bigg(\prod_{t=1}^n\mathcal{P}_{U}(u_t)\bigg) \nonumber \\
&\times \sum_{m_{2}}\mathcal{P}^{\sigma}(m_{2}|m_0,m_1)\times \nonumber \\ &\sum_{v_1^n,v_{2}^n}\mathcal{P}^{\tau_1}(v_1^n|m_0,m_1)\mathcal{P}^{\tau_{2}}(v_{2}^n|m_0,m_{2})\times \nonumber \\&\Bigg[\frac{1}{n}\sum_{t=1}^n c_1(u_t,v_{1,t},v_{2,t})\Bigg] \nonumber\\
&=\sum_{t=1}^n\sum_{m_{2},u_t}\mathcal{P}_{U}(u_t)\mathcal{P}^{\sigma}(m_0,m_{2},u_t,t|m_0,m_1,t)\times \nonumber \\ &\sum_{v_{1,t},v_{2,t}}\mathcal{P}^{\tau_{1,t}}(v_{1,t}|m_0,m_1,t)\mathcal{P}^{\tau_{2,t}}(v_{2,t}|m_0,m_{2},t)\times \nonumber \\&\big[c_1(u_t,v_{1,t},v_{2,t})\big]. \nonumber\\
&=\sum_{m_{2},u_t,t}\mathcal{P}_{U}(u_t)\mathcal{P}^{\sigma}(m_{2},t|m_0,m_1,t)\times \nonumber \\ &\sum_{v_{1,t},v_{2,t}}\mathcal{P}^{\tau_{1,t}}(v_{1,t}|m_0,m_1,t)\mathcal{P}^{\tau_{2,t}}(v_{2,t}|m_0,m_2,t)\times \nonumber \\ &C^{\star}_1(v_{1,t},v_{2,t},m_0,m_1,m_2,t) \nonumber \\
&=\sum_{w_{2},u}\mathcal{P}_{U}(u)\mathcal{P}^{\sigma}(w_0,w_{2}|w_0,w_1) \sum_{v_1,v_2}\mathcal{P}^{\tau_{1,t}}(v_1|w_0,w_1)\times \nonumber \\ &\mathcal{P}^{\tau_{2,t}}(v_2|w_0,w_2)C^{\star}_1(v_1,v_2,w_0,w_1,w_2) \nonumber \\
&=\mathbb{E}_{\mathcal{P}_{U}}[\Psi^{\star}_1(\mathcal{P}^{\tau_{1,t}}_{V_1| W_0W_1},\mathcal{P}^{\tau_{2,t}}_{V_2|W_0W_2},w_0,w_1)].\nonumber
\end{align}
\end{proof}

\begin{lemma} \label{lemma:4}
For all $\sigma$, we have
\begin{align}
&\mathbb{BNE}(\mathcal{P}^{\sigma}_{W_0W_1W_2|U}) = \Big\{(\mathcal{Q}_{V_1|W_0W_1},\mathcal{Q}_{V_2|W_0W_2}), \ \nonumber\\
&
\exists (\tau_1,\tau_2), \; \tau_1 \in BR_1^{\sigma}(\tau_2), \tau_2 \in BR_2^{\sigma}(\tau_1), \nonumber\\ &
\; \; \mathcal{Q}_{V_1|W_0W_1} = \mathcal{P}^{\tau_1}_{V_1|W_0W_1}, \mathcal{Q}_{V_2|W_0W_2} = \mathcal{P}^{\tau_2}_{V_2|W_0W_2}\Big\}.  \label{eq:lemmaBRset}
\end{align}
\end{lemma}

\begin{proof}[Lemma \ref{lemma:4}] Fix $\sigma$ and let
$(\mathcal{Q}_{V_1| W_0W_1}, \mathcal{Q}_{V_2|W_0W_2}) \in \mathbb{BNE}(\mathcal{P}^{\sigma}_{W_0W_1W_2|U})$. We consider $(\tau_1,\tau_2)$ such that 
\begin{align}
    \mathcal{P}^{\tau_1}_{V_1^n|M_0M_1}(v_1^n|m_0,m_1) =& \prod_{t=1}^n \mathcal{Q}_{V_1| W_0W_1}(v_{1,t}|m_0,m_1,t), \nonumber\\
    \mathcal{P}^{\tau_2}_{V_2^n|M_0M_2}(v_2^n|m_0,m_2) =& \prod_{t=1}^n \mathcal{Q}_{V_2| W_0W_2}(v_{2,t}|m_0,m_2,t).\nonumber
\end{align} 

Then $\forall (w_0,w_1,v_1)=(m_0,m_1,t,v_{1,t})$,
\begin{align}
&\mathcal{P}^{\tau_1}_{V_1|W_0W_1}(v_1|w_0,w_1) = \mathcal{P}^{\tau_1}_{V_1|W_0W_1}(v_{1,t}|m_0,m_1,t) \nonumber\\
=& \sum_{v_1^{t-1},v_{1,t+1}^n} \prod_{s=1}^n \mathcal{Q}_{V_1| W_0W_1}(v_{1,s}|m_0,m_1,s)\nonumber\\
=&  \mathcal{Q}_{V_1| W_0W_1}(v_{1,t}|m_0,m_1,t) \nonumber \\ 
&\times 
\sum_{v_1^{t-1},v_{1,t+1}^n} \prod_{s\neq t} \mathcal{Q}_{V_1| W_0W_1}(v_{1,s}|m_0,m_1,s)\nonumber\\
=&  \mathcal{Q}_{V_1| W_0W_1}(v_{1,t}|m_0,m_1,t)= \mathcal{Q}_{V_1| W_0W_1}(v_{1}|w_0,w_1).
\label{eq:T2}
\end{align} 
Assume that $\forall \tau_2, \tau_1 \notin BR_1^{\sigma}(\tau_2)$, then there exists $\bar{\tau}_1 \neq \tau_1$ such that 
\begin{align}
&\E_{\mathcal{P}^{\sigma}_{U}}\big[\Psi^{\star}_1(\mathcal{P}^{\bar{\tau}_1}_{V_1|W_0W_1},\mathcal{P}^{\tau_2}_{V_2|W_0W_2},w_0,w_1)\big] =  \nonumber \\ &\Psi^{\sigma}_1(\bar{\tau}_1,\tau_{2},m_0,m_1)<\Psi^{\sigma}_1(\tau_1,\tau_{2},m_0,m_1)=  \nonumber\\
&\E_{\mathcal{P}^{\sigma}_{U|W_0W_1W_2}}\big[\Psi^{\star}_1(\mathcal{P}^{\tau_1}_{V_1|W_0W_1},\mathcal{P}^{\tau_2}_{V_2|W_0W_2},w_0,w_1)\big], \forall \tau_2, \nonumber
\end{align} 
which contradicts  $(\mathcal{Q}_{V_1| W_0W_1},\mathcal{Q}_{V_2| W_0W_2})\in \mathbb{BNE}(\mathcal{P}^{\sigma}_{W_0W_1W_2|U})$. Therefore, $\tau_1\in BR_1^{\sigma}(\tau_2)$ and thus $\mathcal{Q}_{V_1|W_0W_1}$ belongs to the right-hand side of \eqref{eq:lemmaBRset}. 

Therefore for all $m_0,m_1$,
\begin{align}
&\Psi^{\sigma}_1(\tau_1,\tau_{2},m_0,m_1) \nonumber \\ =& \E_{\mathcal{P}^{\sigma}_{U}}\big[\Psi^{\star}_1(\mathcal{P}^{\tau_1}_{V_1|W_0W_1},\mathcal{P}^{\tau_2}_{V_2|W_0W_2},w_0,w_1)\big] \nonumber \\
=& \min_{\mathcal{Q}_{V_1| W_0W_1}} \E_{\mathcal{P}^{\sigma}_{U}}\big[\Psi^{\star}_1(\mathcal{Q}_{V_1|W_0W_1},\mathcal{Q}_{V_2|W_0W_2},w_0,w_1)\big]\nonumber\\
\leq& \min_{\tilde{\tau}_1} \E_{\mathcal{P}^{\sigma}_{U} }\big[\Psi^{\star}_1(\mathcal{P}^{\tilde{\tau}_1}_{V_1| W_0W_1},\mathcal{P}^{\tau_2}_{V_2|W_0W_2},w_0,w_1)\big] \nonumber \\ =& \min_{\tilde{\tau}_1} \Psi^{\sigma}_1(\tilde{\tau}_1,\tau_{2},m_0,m_1).
\nonumber
\end{align}
Hence $\tau_1\in BR_1^{\sigma}(\tau_2)$.
Similarly,
$\tau_2\in BR_2^{\sigma}(\tau_1)$. The other inclusion is direct. 
\end{proof}

For any strategy $\sigma$, we have
\begin{align}
&\underset{\tau_1,\tau_2}{\max} \;  c_e^n(\sigma,\tau_1,\tau_2) 
=\underset{\tau_1,\tau_2}{\max} \; \mathbb{E}_{\mathcal{P}^{\sigma}_{W_0W_1W_2|U}\atop\mathcal{P}^{\tau_1}_{V_{1}|W_0W_1}  \mathcal{P}^{\tau_2}_{V_{2}|W_0W_2}}\Big[c_e(U,V_{1},V_{2})\Big] \label{zachi1}\\ 
=&\underset{(\mathcal{Q}_{V_{1}|W_0W_1},\mathcal{Q}_{V_{2}|W_0W_2}) \in \atop \mathbb{BNE}(\mathcal{P}^{\sigma}_{W_0W_1W_2|U})}{\max}\mathbb{E}_{\mathcal{P}^{\sigma}_{W_0W_1W_2|U} \atop \mathcal{Q}_{V_{1}|W_0W_1} \ \mathcal{Q}_{V_{2}|W_0W_2}}\Big[c_e(U,V_{1},V_{2})\Big] \label{zachi2}\\ 
\geq&\underset{\mathcal{Q}_{W_0W_1W_2|U}\atop\in\Hat{\mathbb{Q}}_0(R_0,R_1,R_2)}{\inf}\underset{(\mathcal{Q}_{V_{1}|W_0W_1},\mathcal{Q}_{V_{2}|W_0W_2}) \in \atop \mathbb{BNE}(\mathcal{Q}_{W_0W_1W_2|U})}{\max}\mathbb{E}
\Big[c_e(U,V_{1},V_{2})\Big] \label{zachi3}\\ 
=&\Hat{\Gamma}_e(R_0,R_1,R_2). \label{optdistoooo}
\end{align} 
Equations \eqref{zachi1} and \eqref{zachi2} follow from Lemma \ref{lemma:singlelettercost}, whereas \eqref{zachi3} comes from Lemma \ref{lemma:belongtoQ0}. This concludes the converse proof of Theorem \ref{main result}.

\appendices

\section{Proof of Lemma \ref{lemm134} }
\label{battod1}

In the following, we prove Lemma \ref{lemm134}. We control the beliefs of decoder $\mathcal{D}_2$ about the type of decoder $\mathcal{D}_1$. The analogous case can be proven following similar arguments. 
    
We denote the codebook by $\mathcal{C}$, and for each $m_0$, we denote the inner codebook by $\mathcal{C}(m_0)=\{w_1^n(m_0,m_1),w_2^n(m_0,m_2) \in \mathcal{C}, m_1\in\{1,..2^{\lfloor nR_1\rfloor}\}, m_2\in\{1,..2^{\lfloor nR_2\rfloor}\}\}
$. For each $m_0\in\{1,..2^{\lfloor nR_0\rfloor}\}$, $w_1\in\mathcal{W}_1$,  $w_2\in\mathcal{W}_2$, $t\in\{1,\ldots,n\}$ and $\delta>0$, we introduce the following sets
\begin{align}
\mathcal{A}_1^{t}(w_1|m_0) =&\big\{w_1^n \in \mathcal{C}(m_0), w_{1,t}=w_1\big\},\label{qwased}\\
\mathcal{A}_2^{t}(w_2|m_0) =&\big\{w_2^n \in \mathcal{C}(m_0), w_{2,t}=w_2\big\},\label{qwased2}\\
\mathcal{J}(m_0)=&\Bigg\{t \in \{1,..,n\}, \nonumber \\ &\bigg|\bigg|\mathcal{Q}_{W_1|W_0}(\cdot|w_{0t})-\frac{|\mathcal{A}_1^{t}(\cdot|m_0)|}{|\mathcal{C}(m_0)|} \bigg|\bigg| > \delta, \nonumber \\ &\bigg|\bigg|\mathcal{Q}_{W_2|W_0}(\cdot|w_{0t})-\frac{|\mathcal{A}_2^{t}(\cdot|m_0)|}{|\mathcal{C}(m_0)|} \bigg|\bigg| > \delta \Bigg\}, \label{asd123}
\end{align}
where $||\cdot||$ denote the $L_1$-norm.
 We have by the law of large numbers,   $\mathcal{P}(|\mathcal{J}(m_0)|>\delta) \underset{n\mapsto\infty} {\longrightarrow}0 \label{asde123}$.
For any $w_1^n$, $w_0^n$,  and sufficiently small $\delta>0$, and sufficiently large $n$ we have
\begin{align}
\mathcal{P}(w_1^n|w_0^n) =&\sum_{u^n \in \mathcal{T}_{\delta}^n(\mathcal{P}_{U|W_0W_1},w_0^n,w_1^n)}\mathcal{P}(w_1^n,u^n|w_0^n) \\=&\sum_{u^n \in \mathcal{T}_{\delta}^n(\mathcal{P}_{U|W_0W_1},w_0^n,w_1^n)} \mathcal{P}(u^n|w_0^n)\mathcal{P}(w_1^n|u^n,w_0^n)  \\ 
=& \sum_{u^n \in \mathcal{T}_{\delta}^n(\mathcal{P}_{U|W_0W_1},w_0^n,w_1^n)}2^{-nH(U|W_0)} \mathbbm{1}(E_{\delta}=0)  \\ 
=&|\mathcal{T}_{\delta}^n(\mathcal{P}_{U|W_0W_1},w_0^n,w_1^n)|2^{-nH(U|W_0)}   \\=&
2^{nH(U|W_0W_1)}2^{-nH(U|W_0)} \\ =&
2^{-nI(U;W_1|W_0)} \\ 
=& 
2^{-n(R_1-\eta)}. \label{82} 
\end{align}

For any $ t, w_{1t}, w_0^n$, and sufficiently small $\delta>0$, we have
\begin{align}
 \mathcal{P}_{W_{1t}|W_0^n}(w_{1t}|w_0^n) &= 
\sum_{w_1^n \in \mathcal{T}^{n}_{\delta}(P_{W_1|W_0},w^{n}_0)}\mathcal{P}(w_1^n|w_0^n) \\
&=\sum_{w_1^n \in \mathcal{A}_1^{t}(w_{1t}|w^{n}_0)}2^{-n(R_1-\eta)}
\label{eqzxc84}\\
  &=
|\mathcal{A}_1^{t}(w_{1t}|w_0^{n})|2^{-n(R_1-\eta)}
 \\ 
  &=
\label{86}
\frac{|\mathcal{A}_1^{t}(w_{1t}|w_0^{n})|}{|\mathcal{C}(w_0^n)|}
 \\ 
 &=
 \mathcal{Q}_{W_{1t}|W_{0t}}(w_{1t}|w_{0t}) \label{aa105}
\end{align}
where \eqref{eqzxc84} follows from \eqref{82} 
, and \eqref{aa105} follows from \eqref{asd123}. Therefore, we get
\begin{align}
     &\lim_{n\mapsto\infty}\mathbb{E} \Big[\frac{1}{n}\sum_{t=1}^n D(\mathcal{P}_{W_{1t}|W^n_0}(.|W_0^n) || \mathcal{Q}_{W_{1t}|W_{0t}}(.|W_{0t})) \Big|E_{\delta}=0\Big] \\
     =&\lim_{n\mapsto\infty}\sum_{w_1^n,w_0^n}\mathcal{P}^{\sigma}(w_1^n,w_0^n\Big| E_{\delta}=0) \frac{1}{n}\sum_{t=1}^n\sum_{w_{1}} \mathcal{P}(w_{1}|w^n_0)\times \nonumber \\ &\log_2\frac{\mathcal{P}(w_{1}|w^n_0)}{\mathcal{Q}(w_{1t}|w_{0t})} \\
     =& 0 \label{89sd}.
\end{align}
where \eqref{89sd} follows since $\underset{n\mapsto\infty}{\lim}\log\frac{\mathcal{P}_{W_{1t|W_0^n}}(w_{1t}|w_0^n)}{\mathcal{Q}_{W_{1t}|W_{0t}}(w_{1t}|w_{0t})} =0$.

\section{Control of Beliefs about the state $U$: Proof of Lemma \ref{lemm133}} 
 \label{cobatsu}

We denote the Bayesian posterior belief about the state $\mathcal{P}^{\sigma}_{U_t|M_1M_2M_0}(\cdot|m_1,m_2,m_0)\in\Delta(\mathcal{U})$ 
by $\mathcal{P}^{m_1,m_2,m_0}_{U_t}$ 
. We show that on average, the Bayesian belief is close in KL distance to the target belief $\mathcal{Q}_{U|W_{0}W_{1}W_{2}}$ 
induced by the single-letter distribution $\mathcal{Q}_{W_{0}W_{1}W_{2}|U}$. 
The indicator of error event $E_{\delta} \in \{0,1\}$ is as given in \eqref{yyttyytt}. 
Assuming the distribution $\mathcal{Q}_{U|W_1W_2W_0}$ is fully supported, the beliefs about the state are controlled as follows
  \begin{align}
       &\mathbb{E} \Big[\frac{1}{n}\sum_{t=1}^n D(\mathcal{P}^{m_1,m_2,\atop m_0}_{U_t} ||\mathcal{Q}_{U}(\cdot | W_{1t},W_{2t},W_{0t})) \Big|E_{\delta}=0\Big] \\
    =& \sum_{m_1,m_2,m_0,\atop w^n_2 
, w_1^n,w_0^n}\mathcal{P}^{\sigma,\tau_1,\tau_2}(m_1,m_2,m_0,w_1^n, w^n_2 ,w_0^n\Big| E_{\delta}=0)\nonumber \times\\ &\frac{1}{n}\sum_{t=1}^n D(\mathcal{P}^{m_1,m_2,\atop m_0}_{U_t} ||\mathcal{Q}_{U|W_{0}W_{1}W_{2}}(\cdot | W_{1t},W_{2t},W_{0t})) \label{apb2}\\ 
     =& \sum_{m_1,m_2,m_0,\atop w^n_2 
, w_1^n,w_0^n}\mathcal{P}^{\sigma,\tau_1,\tau_2}(m_1,m_2,m_0,w_1^n, w^n_2 ,w_0^n\Big| E_{\delta}=0) \nonumber \times \\ &\frac{1}{n}\sum_{t=1}^n\sum_{u}\mathcal{P}_{U_t}^{m_1,m_2,\atop m_0}(u)\log_2\frac{\mathcal{P}_{U_t}^{m_1,m_2,\atop m_0}(u)}{\mathcal{Q}_{U|W_{0}W_{1}W_{2}}(u| w_{1t},w_{2t},w_{0t})} \label{apzbb2}\\
     =& \sum_{m_1,m_2,m_0,\atop w^n_2 
, w_1^n,w_0^n}\mathcal{P}^{\sigma,\tau_1,\tau_2}(m_1,m_2,m_0,w_1^n,w^n_2 ,w_0^n\Big| E_{\delta}=0)\times \nonumber \\ &\frac{1}{n}\sum_{t=1}^n\sum_{u}\mathcal{P}_{U_t}^{m_1,m_2,\atop m_0}(u)\log_2 \frac{1}{\mathcal{Q}_{U|W_{0}W_{1}W_{2}}(u| w_{1t},w_{2t},w_{0t})} \nonumber \\
     &- \sum_{m_1,m_2,m_0,\atop w^n_2 
, w_1^n,w_0^n}\mathcal{P}^{\sigma,\tau_1,\tau_2}(m_1,m_2,m_0,w_1^n,w^n_2 ,w_0^n\Big| E_{\delta}=0)\nonumber \times \\ &\frac{1}{n}\sum_{t=1}^n\sum_{u}\mathcal{P}_{U_t}^{m_1m_2,\atop m_0}(u)\log_2\frac{1}{\mathcal{P}_{U_t}^{m_1,m_2,\atop m_0}(u)} \label{apzbbb2}\\
    =& \frac{1}{n}  \sum_{m_1,m_2,m_0,\atop w^n_2 
, w_1^n,w_0^n}\mathcal{P}^{\sigma,\tau_1,\tau_2}(m_1,m_2,m_0,w_1^n,w^n_2 ,w_0^n\Big| E_{\delta}=0)\nonumber \times \\ &\sum_{t=1}^n\sum_{u}\mathcal{P}_{U_t}^{m_1m_2m_0}(u)\log_2\frac{1}{\mathcal{Q}_{U|W_{0}W_{1}W_{2}}(u| w_{1t},w_{2t},w_{0t})} \nonumber \\ & -\frac{1}{n}\sum_{t=1}^n H(U_t|M_1,M_2,M_0, E_{\delta}=0) \label{apb3} \\
    =& \frac{1}{n}\sum_{u^n,w_1^n,w^n_2 ,w_0^n}\mathcal{P}^{\sigma,\tau_1,\tau_2}(u^n,w_1^n, w^n_2 ,w_0^n\Big| E_{\delta}=0)\nonumber \times \\ &\log_2\frac{1}{\Pi_{t=1}^n\mathcal{Q}_{U|W_{0}W_{1}W_{2}}(u_t| w_{1t},w_{2t},w_{0t})} -\nonumber \\ &\frac{1}{n}\sum_{t=1}^n H(U_t|M_1,M_2,M_0, E_{\delta}=0) \label{apb4}\\
    =& \frac{1}{n}\sum_{u^n,w_1^n,w^n_2,w_0^n \in \mathcal{T}_{\delta}^n}\mathcal{P}^{\sigma,\tau_1,\tau_2}(u^n,w_1^n, w^n_2,w_0^n\Big| E_{\delta}=0)\nonumber \times \\ &\log_2\frac{1}{\Pi_{t=1}^n\mathcal{Q}_{U|W_{0}W_{1}W_{2}}(u_t| w_{1t},w_{2t},w_{0t})} -\nonumber \\ &\frac{1}{n}\sum_{t=1}^n H(U_t|M_1,M_2,M_0, E_{\delta}=0) \label{apb4444}\\
      \leq& \frac{1}{n}\sum_{u^n,w_1^n,\atop  w^n_2 ,w_0^n\in \mathcal{T}_{\delta}^n}\mathcal{P}^{\sigma,\tau_1,\tau_2}(u^n,w_1^n, w^n_2 ,w_0^n\Big| E_{\delta}=0)\cdot n \times \nonumber \\ &\big(H(U|W_1,W_2,W_0) + \delta \big) \nonumber \\ -&\frac{1}{n}H(U^n|M_1,M_2,M_0, E_{\delta}=0)  \label{apb5}\\
    \leq& \frac{1}{n}I(U^n;M_1,M_2,M_0\Big| E_{\delta}=0)-I(U;W_1,W_2,W_0)+\delta \nonumber \\ &+\frac{1}{n}+ \log_2|\mathcal{U}|\cdot\mathcal{P}^{\sigma,\tau_1,\tau_2}(E_{\delta}=1) \label{apb6}\\
   \leq& \eta + \delta +\frac{1}{n}+\log_2|\mathcal{U}|\cdot\mathcal{P}^{\sigma,\tau_1,\tau_2}(E_{\delta}=1)\label{apb7}. 
\end{align}
\begin{itemize}
    \item Equation \eqref{apb2} comes from the definition of expected K-L divergence. 
     \item Equation \eqref{apzbb2} comes from the definition of K-L divergence. 
     \item Equation \eqref{apzbbb2} comes from splitting the logarithm.  
    \item Equation \eqref{apb3} follows since: 
    \begin{align}
     &\sum_{m_1,m_2,m_0,\atop w^n_2 
, w_1^n,w_0^n}\mathcal{P}^{\sigma,\tau_1,\tau_2}(m_1,m_2,m_0,w_1^n,w_2^n,w_0^n\Big| E_{\delta}=0)\times \nonumber \\ &\frac{1}{n}\sum_{t=1}^n\sum_{u}\mathcal{P}_{U_t}^{m_1m_2m_0}(u)\log_2\frac{1}{\mathcal{P}_{U_t}^{m_1m_2m_0}(u)} \label{apbb2} \\
     =& \sum_{m_1,m_2,m_0,\atop w^n_2 
, w_1^n,w_0^n}\mathcal{P}^{\sigma,\tau_1,\tau_2}(m_1,m_2,m_0,w_1^n,w_2^n,w_0^n\Big| E_{\delta}=0)\times \nonumber \\ &\frac{1}{n}\sum_{t=1}^n H(U_t|M_1=m_1,M_2=m_2,M_0=m_0) \\ 
     =&\frac{1}{n}\sum_{t=1}^n\sum_{m_1,m_2,m_0,\atop w_2^n 
, w_1^n,w_0^n}\mathcal{P}(m_1,m_2,m_0,w_1^n,w_2^n,w_0^n\Big| E_{\delta}=0)\nonumber \\ &\times H(U_t|M_1=m_1,M_2=m_2,M_0=m_0) \\ 
     =&\frac{1}{n}\sum_{t=1}^n\sum_{m_1,m_2,m_0}\mathcal{P}^{\sigma,\tau_1,\tau_2}(m_1,m_2,m_0\Big| E_{\delta}=0)\times \nonumber \\&H(U_t|M_1=m_1,M_2=m_2,M_0=m_0) \\ 
     =& \frac{1}{n}\sum_{t=1}^n H(U_t|M_1,M_2,M_0,E_{\delta}=0).
\end{align}

    \item Equation \eqref{apb4} follows since: \begin{align}
    &\sum_{m_1,m_2,m_0,\atop w_2 
, w_1^n,w_0^n}\mathcal{P}^{\sigma,\tau_1,\tau_2}(m_1,m_2,m_0,w_1^n,w_0^n\Big|E_{\delta}=0)\times \nonumber \\ &\frac{1}{n}\sum_{t=1}^n\sum_{u}\mathcal{P}_{U_t}^{m_1m_2m_0}(u)\log_2\frac{1}{\mathcal{Q}_{U|W_{0}W_{1}W_{2}}(u| w_{1t},w_{2t},w_{0t})} \\
    =&\frac{1}{n}\sum_{t=1}^n\sum_{u_t,m_1,m_2,\atop m_0,w_1^n,w_0^n}\mathcal{P}^{\sigma,\tau_1,\tau_2}(u_t,m_1,m_2, m_0,w_1^n,w_0^n\Big|E_{\delta}=0)\nonumber \\ &\times \log_2\frac{1}{\mathcal{Q}_{U|W_{0}W_{1}W_{2}}(u_t| w_{1t},w_{2t},w_{0t})} \\
        =&\frac{1}{n}\sum_{t=1}^n\sum_{u^n,m_1,m_2,\atop m_0,w_1^n,w_0^n}\mathcal{P}^{\sigma,\tau_1,\atop \tau_2}(u^n,m_1,m_2,m_0, w_1^n,w_0^n\Big|E_{\delta}=0) \nonumber \\ &\times\log_2\frac{1}{\mathcal{Q}_{U|W_{0}W_{1}W_{2}}(u_t| w_{1t},w_{2t},w_{0t})} \\
        =&\frac{1}{n}\sum_{u^n,m_1,m_2,\atop m_0, w_1^n,w_0^n}\mathcal{P}^{\sigma,\tau_1,\tau_2}(u^n,m_1,m_2,m_0,w_1^n,w_2^n,w_0^n\Big|E_{\delta}=0)\nonumber \\ &\times \log_2\frac{1}{\Pi_{t=1}^n\mathcal{Q}_{U|W_{0}W_{1}W_{2}}(u_t| w_{1t},w_{2t},w_{0t})} \\
        =&\frac{1}{n}\sum_{u^n,w_1^n,w_2^n,w_0^n}\mathcal{P}^{\sigma,\tau_1,\tau_2}(u^n,w_1^n,w_2^n,w_0^n\Big|E_{\delta}=0)\nonumber \\ &\times \log_2\frac{1}{\Pi_{t=1}^n\mathcal{Q}_{U|W_{0}W_{1}W_{2}}(u_t| w_{1t},w_{2t},w_{0t})}.
\end{align}

\item Equation \eqref{apb4444} follows since the support of $\mathcal{P}^{\sigma,\tau_1,\tau_2}(u^n,w_1^n,w_2^n,w_0^n|E_{\delta})=\mathbb{P}\{(u^n,w_1^n,w_2^n,w_0^n) \in \mathcal{T}_{\delta}^n \}$ is included in $\mathcal{T}_{\delta}^n$. 
    \item Equation \eqref{apb5} follows from the typical average lemma property (Property 1 pp.26 in \cite{elgamal}) given in lemma \ref{apa.20jet}, and the chain rule of entropy: $H(U^n|M_1,M_2,M_0,W^n_1,W^n_2,W_0^n) \leq \sum_{t=1}^n H(U_t|M_1,M_2,M_0,W_1,W_2,W_0)$. 
   \item Equation \eqref{apb6} comes from the conditional entropy property and the fact that $H(U^n)=nH(U)$ for an i.i.d random variable $U$ and lemma \ref{apa.2222jet}.
    \item Equation \eqref{apb7} follows since $I(U^n;M_1,M_2,M_0) \leq H(M_1,M_2,M_0) \leq \log_2 |J| = n \cdot (R_1+R_2+R_0) = n \cdot (I(U;W_1,W_2,W_0)+\eta)$ and lemma \ref{apa.2222jet}.
\end{itemize}

If the expected probability of error is small over the codebooks, then it has to be small over at least one codebook. Therefore, equations  \eqref{eqarr} and \eqref{eqarrr} imply that: \begin{align}
    &\forall \varepsilon_2>0, \forall \eta>0, \exists \Bar{\delta}>0,\forall \delta \leq \Bar{\delta}, \exists \Bar{n} \in \mathbb{N}, \forall n\geq \Bar{n}, \exists b^{\star}, \nonumber \\  &\mathrm{such that} \ \mathcal{P}_{b^{\star}}(E^2_{\delta}=1)\leq \varepsilon_2. \label{apeqwaa} 
\end{align}     
       The strategy $\sigma$ of the encoder consists of using $b^{\star}$ in order to transmit the pair $(m_1,m_2,m_0)$ such that \\ $(U^n,W^n_0(m_{0}),W_1^n(m_{0},m_1))$ is a jointly typical sequence. By construction, this satisfies equation \eqref{apeqwaa}. 
       
\lemma \label{aplemmmaa10}Let $\mathcal{Q}_{W_{0}W_{1}W_{2}|U} \in \tilde{Q}_{0}(R_1,R_2,R_0)$, then $\forall \varepsilon>0$, $\forall \alpha>0,\gamma>0$, there exists $\Bar{\delta}$, $\forall \delta \leq\Bar{\delta}$, $\exists \Bar{n}$, $\forall n \geq \Bar{n}$, $\exists \sigma$, such that   $1-\mathcal{P}^{\sigma}(B_{\alpha,\gamma,\delta})\leq \varepsilon$. \\
      \proof{of lemma \ref{aplemmmaa10}}
      We have:  \begin{align}
          &1-\mathcal{P}_{\sigma}(B_{\alpha,\gamma,\delta}):=\mathcal{P}_{\sigma}(B^c_{\alpha,\gamma,\delta}) \\
          &= \mathcal{P}_{\sigma}(E_{\delta}=0)\mathcal{P}_{\sigma}(B^c_{\alpha,\gamma,\delta}|E_{\delta}=0) \\
          &\leq \mathcal{P}_{\sigma}(B^c_{\alpha,\gamma,\delta}|E_{\delta}=0) \\
          &\leq \varepsilon_{2} + \mathcal{P}_{\sigma}(B^c_{\alpha,\gamma,\delta}|E_{\delta}=1). \label{apeqw222}
      \end{align}
      \ \ Moreover, \begin{align}
          &\mathcal{P}_{\sigma}(B^c_{\alpha,\gamma,\delta}|E_{\delta}=0) = \nonumber \\ &\sum_{w_1^n,w_2^n,w_0^n,\atop m_1,m_2,m_0}\mathcal{P}_{\sigma}\bigg((w_1^n,w_0^n,m_1,m_2,m_0)\in B^c_{\alpha,\gamma,\delta}\Bigg|E_{\delta}=0\bigg) \label{apeqa1}\\
          &= \sum_{w_1^n,w_2^n,w_0^n,\atop m_1,m_2,m_0}\mathcal{P}_{\sigma}\bigg((w_1^n,w_0^n,m_1,m_2,m_0) \ \mathrm{ s.t. } \nonumber \\ &\frac{|T_{\alpha}(w_1^n,w_0^n,m_1,m_2,m_0)|}{n}\leq 1-\gamma \Bigg|E_{\delta}=0\bigg) \\
          &= \mathcal{P}_{\sigma}\bigg(\frac{\#}{n}\bigg\{t, D\bigg(\mathcal{P}^{m_1,m_2,m_0}_{U_t} \bigg|\bigg|\mathcal{Q}_{U|W_{0}W_{1}W_{2}}(\cdot | W_{1t},W_{2t},W_{0t})\bigg) \nonumber \\ &\leq \frac{\alpha^2}{2\ln2}<1-\gamma\Bigg| E_{\delta}=0\bigg\} \\
          &= \mathcal{P}_{\sigma}\bigg(\frac{\#}{n}\bigg\{t, D\bigg(\mathcal{P}^{m_1,m_2,m_0}_{U_t} \bigg|\bigg|\mathcal{Q}_{U|W_{0}W_{1}W_{2}}(\cdot | W_{1t},W_{2t},W_{0t})\bigg) \nonumber \\ &> \frac{\alpha^2}{2\ln2} \nonumber \\ &\geq \gamma\Big| E_{\delta}=0\bigg\} \label{apeqa4}\\
          &\leq \frac{2\ln{2}}{\alpha^2\gamma}\cdot\mathbb{E}_{\sigma}\bigg[\frac{1}{n}\sum_{t=1}^n D\bigg(\mathcal{P}^{m_1,m_2,\atop m_0}_{t} \bigg|\bigg|\mathcal{Q}_{U}(\cdot | W_{1t},W_{2t},W_{0t})\bigg)\bigg] \label{apeqa5}\\
          &\leq \frac{2\ln{2}}{\alpha^2\gamma}\cdot \bigg(\eta +\delta+\frac{2}{n}+2\log_2|\mathcal{U}|\cdot\mathcal{P}_{\sigma}(E^2_{\delta}=1),\bigg) \label{apeqa6}
      \end{align}
      \begin{itemize}
          \item Equations \eqref{apeqa1} to \eqref{apeqa4} are simple reformulations. 
          \item Equation \eqref{apeqa5} comes from using Markov's inequality given in lemma \ref{apeqa0}.  
          \item Equation \eqref{apeqa6} comes from equation \eqref{apb7}. 
      \end{itemize}
\lemma (Markov's Inequality)\label{apeqa0}. For all $\varepsilon_1>0$ , $\varepsilon_2>0$ we have:
\begin{align}
    &\mathbb{E}_{\sigma}\bigg[\frac{1}{n}\sum_{t=1}^n D\bigg(\mathcal{P}^{m_1,m_2,m_0}_{U_t} \bigg|\bigg|\mathcal{Q}_{U|W_{0}W_{1}W_{2}}(\cdot | W_{1t},W_{2t},W_{0t})\bigg)\bigg] \nonumber \\ &\leq \varepsilon_0 
    \implies \nonumber \\ &\mathcal{P}\bigg(\frac{\#}{n}\bigg\{t, D\bigg(\mathcal{P}^{m_1,m_2,m_0}_{U_t} \bigg|\bigg|\mathcal{Q}(\cdot | W_{1t},W_{2t},W_{0t})\bigg)>\varepsilon_1\bigg\}>\varepsilon_2\bigg) \nonumber \\ &\leq \frac{\varepsilon_0}{\varepsilon_1\cdot\varepsilon_2}.
\end{align} 
      \proof{of lemma \ref{apeqa0}} We denote by $D_t=D(\mathcal{P}^{m_1,m_2,m_0}_{U_t} ||\mathcal{Q}_{U|W_{0}W_{1}W_{2}}(\cdot | W_{1t},W_{2t},W_{0t})$ and $D^n=\{D_t\}_t$ the K-L divergence. We have that:
      \begin{align}
          &\mathcal{P}\bigg(\frac{\#}{n}\bigg\{t, \mathrm{s.t.} D_t>\varepsilon_1\bigg\}>\varepsilon_2\bigg) \nonumber \\ &=\mathcal{P}\bigg(\frac{1}{n}\cdot\sum^n_{t=1}\mathbbm{1}\bigg\{D_t>\varepsilon_1\bigg\}>\varepsilon_2 \bigg) \label{apmarkovchain1}\\
          \leq& \frac{\mathbb{E}\bigg[\frac{1}{n}\cdot\sum_{t=1}^n\mathbbm{1}\bigg\{D_t>\varepsilon_1\bigg\}\bigg]}{\varepsilon_2} \label{apmarkovchain2}\\
          =& \frac{\frac{1}{n}\sum_{t=1}^n\mathbb{E}\bigg[\mathbbm{1}\bigg\{D_t>\varepsilon_1\bigg\}\bigg]}{\varepsilon_2} \label{apmarkovchain3}\\
          =& \frac{\frac{1}{n}\sum_{t=1}^n\mathcal{P}\bigg(D_t>\varepsilon_1\bigg)}{\varepsilon_2} \label{apmarkovchain4} \\
          \leq& \frac{\frac{1}{n}\sum_{t=1}^n\frac{\mathbb{E}\bigg[D_t\bigg]}{\varepsilon_1}}{\varepsilon_2} \label{apmarkovchain5}\\
          =& \frac{1}{\varepsilon_1\cdot\varepsilon_2}\cdot \mathbb{E}\bigg[\frac{1}{n}\sum^n_{t=1}D_t\bigg] \leq \frac{\varepsilon_0}{\varepsilon_1\cdot\varepsilon_2}. \label{apmarkovchain6}
      \end{align}
       \begin{itemize}
           \item Equations \eqref{apmarkovchain1},  \eqref{apmarkovchain3}, \eqref{apmarkovchain4} and \eqref{apmarkovchain6} are reformulations of probabilities and expectations.
           \item Equations \eqref{apmarkovchain2} and \eqref{apmarkovchain5}, come from Markov's inequality $\mathcal{P}(X\geq \alpha)\leq\frac{\mathbb{E}[X]}{\alpha}, \ \forall \alpha>0$.
       \end{itemize}

\section{Proof of Lemma \ref{lemmtre}} 
 \label{cobatru}
For all $w^n_{0},w^n_{1},w^n_2,t$, we have
\begin{align}
    &\mathcal{P}^{\sigma}(w_{1,t}|w_{0,t},w_{2,t})=\sum_{w^{t-1}_0,w_{0,t+1}^n, \atop w^{t-1}_2,w_{2,t+1}^n}\mathcal{P}^{\sigma}(w^{t-1}_0,w_{0,t+1}^n, \atop w^{t-1}_2,w_{2,t+1}^n|w_{0,t},w_{2,t})\cdot\mathcal{P}^{\sigma}(w_{1,t}|w_0^n,w_2^n). 
\end{align}
Similarly,
\begin{align}
    &\mathcal{P}^{\sigma}(w_{2,t}|w_{0,t},w_{1,t})=\sum_{w^{t-1}_0,w_{0,t+1}^n, \atop w^{t-1}_1,w_{1,t+1}^n}\mathcal{P}^{\sigma}(w^{t-1}_0,w_{0,t+1}^n, \atop w^{t-1}_1,w_{1,t+1}^n|w_{0,t},w_{1,t})\cdot\mathcal{P}^{\sigma}(w_{2,t}|w_0^n,w_1^n).
\end{align}
Therefore, using Lemma \ref{lemm134} the result follows.

\section{Additional Lemmas}
\lemma{(Typical Sequences Property 1, pp.26 in \cite{elgamal})}.
The typical  sequences $(u^n,w_1^n,w_2^n) \in \mathcal{T}_{\delta}^n$ satisfy: \begin{align}
    \forall \varepsilon>0, \  \exists \Bar{\delta}>0, \  \forall \delta \leq \Bar{\delta}, \  \forall n, \  \forall (u^n,w_1^n,w_2^n) \in \mathcal{T}_{\delta}^n, \nonumber \\ \Bigg|\frac{1}{n}\cdot\log_2\frac{1}{\Pi_{t=1}^n\mathcal{P}(u|w_{1,t},w_{2,t})}-H(U|W_1,W_2) \Bigg|\leq \varepsilon, \end{align}   where  $\Bar{\delta}=\varepsilon\cdot H(U|W_1,W_2)$.

\label{apa.20jet} 

\lemma \label{apa.2222jet} Let $U^n$ an i.i.d random variable and $M$ a random variable. For all $\varepsilon>0$, there exists $\Bar{n}\in \mathbb{N}$, such that for all $n \geq \Bar{n}$, we have  \begin{align}
    H(U^n|E_{\delta}=0) \geq n\cdot\Big(H(U)-\varepsilon\Big).
\end{align}
\begin{proof} 
\begin{align}
    &H(U^n|E_{\delta}=0)= \nonumber \\ &\frac{1}{\mathcal{P}(E_{\delta}=0)}\cdot\Big(H(U^n|E_{\delta}=1)-\mathcal{P}(E_{\delta}=1)\cdot H(U^n|E_{\delta}=1)\Big) \label{ttrq1} \\ 
    &\geq H(U^n|E_{\delta})-\mathcal{P}(E_{\delta}=1)\cdot H(U^n|E_{\delta}=1)\Big) \label{ttrq2} \\ 
    &\geq H(U^n)-H(E_{\delta})-\mathcal{P}(E_{\delta}=1)\cdot H(U^n|E_{\delta}=1)\Big) \label{ttrq3} \\ 
    &\geq H(U^n)-n\cdot \varepsilon \label{ttrq4}.
\end{align} 
\end{proof}

\begin{itemize}
    \item Equation \eqref{ttrq1} follows from the conditional entropy definition.
    \item Equation \eqref{ttrq2} follows since $\mathcal{P}(E_{\delta}=0) \leq 1$.
    \item Equation \eqref{ttrq3} comes from the property $H(U^n|M,E_{\delta})=H(U^n,M,E_{\delta})-H(M) - H(E_{\delta}) \geq H(U^n) -H(M) - H(E_{\delta})$. 
    \item Equation \eqref{ttrq4} follows since $U$ is i.i.d and the definition of $E_{\delta}=1$.Hence, for all $\varepsilon$, there exists an $\Bar{n}\in \mathbb{N}$ such that for all $n \geq \Bar{n}$ we have $H(\mathcal{P}(E_{\delta}=1))+H(M)+\mathcal{P}(E_{\delta}=1)\cdot \log_2|\mathcal{U}| \leq \varepsilon.$ 
\end{itemize}

\bibliographystyle{IEEEtran}
\bibliography{sample}
\end{document}